\documentclass[10pt,conference,letterpaper]{IEEEtran}

\usepackage{amsfonts,color}
\usepackage{amsthm,amsmath,amssymb}
\usepackage{algorithm} 
\usepackage[noend]{algorithmic} 
\usepackage{soul} 
\usepackage{url}
\usepackage{cite}
\usepackage[bookmarks=false,hidelinks]{hyperref}

\usepackage{color}
\usepackage{mathtools}

\usepackage{thmtools, thm-restate}

\usepackage{indentfirst}

{}

\newcommand{\textredTech}[1]{{\color{black}{#1}}}
\newcommand{\textblue}[1]{{\color{black}{#1}}}

\newcommand{\sdollar}[1]{{#1}}
\newcommand{\IM}[1]{\cite{Kempe2003Maximizing,chen2010scalable,chen2009efficient,goyal2011data,jung2012irie,leskovec2007cost,nguyen2016cost,nguyen2016stop,ohsaka2014fast,song2015influence,tang2018online,tang2017influence,tang2015influence,tang2014influence,zhou2013ublf}}
\newcommand{\PM}[1]{\cite{tang2016profit,tang2017profit,tang2018towards,zhu2018host,zhang2016profit}}

\newtheorem{lemma}{Lemma}
\theoremstyle{definition}

\newtheorem{observation}{Observation}
\newtheorem{Theorem}{Theorem}
\newtheorem{example}{Example}

\newtheorem*{timeC}{Time complexity}

\newcommand{\problem}{S$^3$CRM}
\newcommand{\algo}{S$^3$CA}

\def\BibTeX{{\rm B\kern-.05em{\sc i\kern-.025em b}\kern-.08em
    T\kern-.1667em\lower.7ex\hbox{E}\kern-.125emX}}

\begin{document}
%
\title{Seed Selection and Social Coupon Allocation for Redemption Maximization in Online Social Networks}


\author{\IEEEauthorblockN{Tung-Chun Chang\IEEEauthorrefmark{1}\IEEEauthorrefmark{3},
Yishuo Shi\IEEEauthorrefmark{2}, De-Nian Yang\IEEEauthorrefmark{2} and
Wen-Tsuen Chen\IEEEauthorrefmark{3}}
\IEEEauthorblockA{
\IEEEauthorrefmark{1}Research Center for Information Technology Innovation, Academia Sinica, Taiwan\\
\IEEEauthorrefmark{2}Institute of Information Science, Academia Sinica, Taiwan\\
\IEEEauthorrefmark{3}Department of Computer Science, National Tsing Hua University, Taiwan\\
Email: tcchang@citi.sinica.edu.tw, yishuoempty@outlook.com, dnyang@iis.sinica.edu.tw, wtchen@cs.nthu.edu.tw}}


%


\maketitle

\begin{abstract}
Online social networks have become the medium for efficient viral marketing exploiting social influence in information diffusion. However, the emerging application {\em Social Coupon (SC)} incorporating social referral into coupons cannot be efficiently solved by previous researches which do not take into account the effect of SC allocation. The number of allocated SCs restricts the number of influenced friends for each user. In the paper, we investigate not only the seed selection problem but also the effect of SC allocation for optimizing the redemption rate which represents the efficiency of SC allocation. Accordingly, we formulate a problem named {\em Seed Selection and SC allocation for Redemption Maximization ({\problem})} and prove the hardness of {\problem}. We design an effective algorithm with a performance guarantee, called Seed Selection and Social Coupon allocation algorithm. For {\problem}, we introduce the notion of marginal redemption to evaluate the efficiency of investment in seeds and SCs. Moreover, for a balanced investment, we develop a new graph structure called guaranteed path, to explore the opportunity to optimize the redemption rate. Finally, we perform a comprehensive evaluation on our proposed algorithm with various baselines. The results validate our ideas and show the effectiveness of the proposed algorithm over baselines.
\end{abstract}


%
\IEEEpeerreviewmaketitle

\section{Introduction}

Online social networks (OSNs), e.g., Facebook, Twitter, have fostered efficient diffusion of information and advertisement. Social influence \cite{social-influence} is the cornerstone of viral marketing and draws a wide spectrum of research, such as influence maximization (IM) \cite{Kempe2003Maximizing,chen2010scalable,chen2009efficient,goyal2011data,jung2012irie,leskovec2007cost,nguyen2016cost,nguyen2016stop,ohsaka2014fast,song2015influence,tang2018online,tang2017influence,tang2015influence,tang2014influence,zhou2013ublf} (identifying the $k$ seed users for maximizing the total influence spread) and profit maximization (PM) \cite{tang2017profit,tang2018towards,zhu2018host,zhang2016profit,khan2016revenue,lu2012profit,zhu2013influence} (maximizing the total benefit from influenced users subtracted by the total expense). Recently, {\em Social Coupon (SC)} \cite{businessinsider,hanson2017friends}, emerges to incorporate social referral into coupons \cite{businessinsider,hanson2017friends}. Real examples include Dropbox\footnote{\url{https://www.dropbox.com/referrals}}, which offers free space to a user that shares cost-effective deals to Facebook friends. The user will acquire 500 MB free space after any friend accepts the deal (up to 16 GB for each user). Airbnb\footnote{\url{https://www.airbnb.com/invite},~\url{https://www.airbnb.com/help/article/2269/airbnb-referral-program-terms-and-conditions}} allows users to invite their Facebook friends to sign up and completes the first trip, whereas \$29 and \$18 travel credits will be rewarded to the invitee and the inviter (up to \$5000 for each user), respectively. Similarly, Booking.com\footnote{\url{https://www.booking.com/content/referral-faq.en-gb.html}} encourages social referral (up to 10 friends for each user). To efficiently facilitate SC, new start-up companies (e.g., ReferralCandy and Extole\footnote{\url{https://www.referralcandy.com},~\url{https://www.extole.com}}) develop a customized SC system for e-commerce websites of clients. However, SC has drawn much less research attention, and currently the efficiency of SC is usually low (0.8\% and 15.9\% for offline and online SC), because resources are not properly allocated to valuable users \cite{jung2010online}.

The research of IM blossoms from Kempe et al. \cite{Kempe2003Maximizing} that introduced the basic formulation, propagation models and the corresponding approximation algorithm, and many follow-up studies improve the efficiency for massive social networks \cite{chen2010scalable,chen2009efficient,goyal2011data,jung2012irie,leskovec2007cost,nguyen2016cost,nguyen2016stop,ohsaka2014fast,song2015influence,tang2018online,tang2017influence,tang2015influence,tang2014influence,zhou2013ublf}. PM \cite{tang2017profit,tang2018towards,zhu2018host,zhang2016profit,khan2016revenue,lu2012profit,zhu2013influence} addresses the trade-off between the seed cost and the benefit of influenced users. However, SC is different to IM and PM due to the following reasons. 1) For SC, not only the seeds but also the internal nodes in an influence spread are required to be selected and allocated resources. In contrast, IM and PM target on only the seeds. 2) For SC, each internal node is associated with an {\em SC constraint}, which limits the maximum number of friends that can be referred and \textit{receive} SC (i.e., influenced by SC and then redeem it). For example, at most 32 (16 GB/500 MB) friends can receive SC from each Dropbox user. On the contrary, each internal node in IM and PM can influence an arbitrary number of friends since no resource is allocated. Therefore, the shape of a spread can be controlled more precisely by SC (e.g., broader or deeper) by carefully selecting internal nodes and allocating different resources to each of them. Moreover, each internal node is more inclined to be activated due to the reward from social referral. 


In this paper, therefore, we make the first attempt to explore {\em seed selection with SC allocation}. Simultaneously selecting and allocating resources to both seeds and internal nodes raise the following new research challenges. The first one is a {\em benefit and total cost trade-off} between the benefit of influenced users and the total cost of seeds and SCs. Unlike previous works, SC needs to allocate difference resources to activate seeds directly and internal nodes through SCs. It is important to carefully examine the benefit of activating users by SCs or delegating them as seeds. The second one is {\em investment trade-off} between the investment in seeds and internal nodes due to a limited budget. Seeds start the influence spread while internal nodes sustain it, and the activation of seeds is usually more expensive. Thus, it is important to carefully allocate the resource to derive an efficient spread. The third one is {\em SC allocation trade-off} between allocating SCs to users near seeds and users with high benefit but far from seeds. Allocating SCs to users near seeds can strengthen the influence spread from its root, i.e., increase the number of expected redeemed SCs, while allocating SCs to high benefit users can improve the efficiency of the investment.

To address the above challenges, we formulate a new optimization problem, named Seed Selection and SC allocation for Redemption Maximization ({\problem}). In contrast to selecting only seeds to maximize the influence \cite{Kempe2003Maximizing,chen2010scalable,chen2009efficient,goyal2011data,jung2012irie,leskovec2007cost,nguyen2016cost,nguyen2016stop,ohsaka2014fast,song2015influence,tang2018online,tang2017influence,tang2015influence,tang2014influence,zhou2013ublf} or profit \cite{tang2017profit,tang2018towards,zhu2018host,zhang2016profit,khan2016revenue,lu2012profit,zhu2013influence}, {\problem} aims to select seeds and internal nodes and allocate SCs to users such that the redemption rate is maximized. The redemption rate is the ratio of the expected benefit of influenced users to the total cost (i.e., the sum of seed cost and SC cost). It is important to choose a proper seed set and internal nodes to start and shape the influence spread. Moreover, it is crucial to balance the investment in seeds and SCs, and meanwhile the total cost must be ensured within the investment budget.

We first prove the NP-hardness for {\problem} and then devise an approximation algorithm, named Seed Selection and SC allocation Algorithm ({\algo}). For the benefit and total cost trade-off, {\algo} introduces {\em marginal redemption} to measure the extra benefit obtained by activating a seed or allocating an SC. For investment trade-off, {\algo} {\em deploys} the resource by iteratively investing a new seed or an SC to an activated user with the highest marginal redemption to balance the investment in seeds and SCs. For SC allocation trade-off, {\algo} first constructs {\em guaranteed paths} to identify all possibly influenced inactive users for each seed under the budget constraint. Moreover, {\algo} introduces {\em amelioration index} and {\em deterioration index} to evaluate the improvement and deterioration of the {\em maneuver} of the previous deployed SCs to optimize the redemption rate. Finally, we evaluate {\algo} with real datasets, and the simulation results show that {\algo} can effectively improve the redemption rate up to 30 times comparing to the baselines.

The rest of this paper is organized as follows. Section \ref{sec:related} reviews the related works. The problem formulation and the hardness analysis are given in Section \ref{sec:problem}. We present the algorithm design in Section \ref{sec:algorithm}, and the performance analysis is described in Section \ref{sec:analysis}. The experimental results are shown in Section \ref{sec:sim}. Finally, Section \ref{sec:conclusion} concludes this paper.

\section{Related Work}
\label{sec:related}
Influence maximization focuses on the pinpoint of the $k$ most influential users in an OSN, and the identified users are delegated to start the influence propagation. The studies of influence maximization blossom from Kempe et al. \cite{Kempe2003Maximizing} introducing the basic formulation and propagation models. Furthermore, submodularity exhibits when the influence diffuses with the models, and thus they proposed a $(1-1/e)$--approximation algorithm. After the algorithm is proposed, many studies followed up to improve the efficiency \cite{chen2010scalable,chen2009efficient,goyal2011data,jung2012irie,leskovec2007cost,nguyen2016cost,nguyen2016stop,ohsaka2014fast,song2015influence,tang2018online,tang2017influence,tang2015influence,tang2014influence,zhou2013ublf} for massive OSNs. However, these works are not applicable to the SC scenario since they consider only seed selection without the SC allocation of internal nodes.

Tang et al. \cite{tang2017profit} state that the optimal profit cannot be derived from influence maximization since the influential users are probably expensive for activation, which results in a trade-off between the expense and the benefit of propagation. Thus profit maximization emerges, and profit is defined as the benefit from influenced users subtracted by the cost of activating seeds in \cite{tang2017profit}. Tang et al. \cite{tang2018towards} take the cost of activating non-seed users into consideration, which is a constant number for each user. However, it is not applicable to the SC scenario where the cost of non-seed users depends on the number of their allocated SCs. Zhu et al. \cite{zhu2018host} and Zhang et al. \cite{zhang2016profit} find the profits of activating users for the competitive environment and multiple products, respectively. Khan et al. \cite{khan2016revenue} formulated a revenue maximization problem from the perspective of the host of OSNs. Furthermore, Lu et al. \cite{lu2012profit} and Zhu et al. \cite{zhu2013influence} maximize the profit, where people are not influenced unless the price meets their expectation. These works \cite{zhu2018host,zhang2016profit,khan2016revenue,lu2012profit,zhu2013influence} are complementary to {\problem}.

\section{OSN Model and Problem Formulation}
\label{sec:problem}

For the SC scenario, independent cascade (IC) model\footnote{Since the SC is usually redeemed solely, the linear threshold is not suitable.} \cite{Kempe2003Maximizing} is extended by imposing an SC constraint to each user as follows. The influence propagation starts from the selected seeds, and then the seeds switch to the active state and try to activate their friends with the influence probability. Moreover, for each $v_i \in V$,  the maximum number of activated friends is restricted by the SC constraint (i.e., no more than $k_i$). The activation starts from the friend with the highest to the lowest influence probability, since in reality the friend with the highest influence probability is most likely to redeem the SC. Each user can be activated only once, and the activation process stops when no more user can be activated. 

\textblue{
For convenience, we have summarized all notations in Table \ref{notation}. The definitions of notations are explained as follows. Users $V$ of the OSN and their relationships $E$ form a weighted directed graph $G=\{V,E\}$, where each edge $e(i,j) \in E$ between $v_i$ and $v_j \in V$ has a weight $P(e(i,j))$, and $b(v_i)$, $c_{seed}(v_i)$, and $c_{sc}(v_i)$ are the benefit, seed cost, and SC cost of $v_i$, respectively. $k_i\in [0, |N(v_i)|] \in \mathbb{N}$ is the SC constraint for each $v_i$, where $|N(v_i)|$ is the number of $v_i$'s friends. Moreover, $S$ and $I$ are the seed set and the internal node set, respectively. $K(I)=\{k_i | \forall v_i \in I\}$ is the SC allocation. Let $C_{seed}(S)$ be the total cost of $S$. Let $C_{sc}(K(I)) = \sum_{v_i \in I, v_j \in N(v_i)} E[k_i, c_{sc}(v_j)]$, where $v_j$ is $v_i$'s friend with the $j$-th highest influence probability and $E[k_i, c_{sc}(v_j)]$ is the expected SC cost to distribute an SC to $v_j$. If $j \leq k_i$, $E[k_i, c_{sc}(v_j)] = c_{sc}(v_j) P(e(i,j))$; otherwise $j > k_i$, $E[k_i, c_{sc}(v_j)] = c_{sc}(v_j) P(e(i,j)) P(\bar{k_i})$ , where $P(\bar{k_i})$ is the probability that at most $k_i-1$ friends $v_{\hat{j}}\in N(v_i), \hat{j} < j$ success to redeem the SC.
Let $B(S,K(I)) = \sum_{v_i \in I} E[S,K(I),b(v_i)]$ be the expected benefit of the influence propagation, where $E[S,K(I),b(v_i)]$ is the expected benefit of $v_i$ with the seed set $S$ and SC allocation $K(I)$. Finally, the investment budget is $B_{inv}$.}

Consider Dropbox as an example. Users {\em distribute} SCs (post the invite link) to their Facebook friends. Once a friend $v_i$ {\em receives} the SC (installs Dropbox through the link), 500 MB of free space is rewarded (i.e., $c_{sc}(v_i)$). For each $v_i \in V$, the maximum number of friends receiving the reward is restricted by the number of {\em allocated} SCs (i.e., $k_i=32$).

\begin{table}[t]
\caption{Notations of {\problem} and {\algo}}%
    \label{notation}
\centering
\footnotesize
\begin{tabular}
[c]{|l|l|}
\hline
\textbf{Notation in Problem} & \textbf{Definition} \\
\hline
$|N(v_i)|$ & number of $v_i$'s friends (out-neighbors) \\
\hline
$c_{seed}(v_i)$, $c_{sc}(v_i)$, $b(v_i)$ & seed cost, SC cost, benefit of $v_i$ \\
\hline
$k_i\in [0, |N(v_i)|] \in \mathbb{N}$  &
SC constraint of $v_i$ \\
\hline
$S \subset V$, $I \subset V$ & seed set, internal node set \\
\hline
$K(I)=\{k_i | \forall v_i \in I\}$ &
 SC allocation  \\
\hline
$C_{seed}(S)$, $C_{sc}(K(I))$ & total seed cost, expected SC cost of $K(I)$ \\ 
\hline
$E[K_i, c_{sc}(v_j)]$  & expected SC cost ($v_i$ distributes an SC to $v_j$) \\
\hline
$B(S,K(I))$ & expected benefit of the influence propagation \\ 
\hline
$E[S,K(I),b(v_i)]$ & expected benefit of $v_i$ with $S$ and $K(I)$ \\
\hline
$B_{inv}$ & investment budget \\
\hline
\textbf{Notation in Algorithm} & \textbf{Definition} \\
\hline
$\tau_{\hat{S},\hat{I}}^{v_i,\gamma_i}$ & MR of selecting $v_i$ as seed (allocate coupon)\\
\hline
$g(s,v_i)$ & the GP from seed $s$ to user $v_i$\\
\hline
$U_s^{{\hat{l}}}$ & ${\hat{l}}$-th level set of visited siblings from seed $s$ \\
\hline
$\hat{K}(I)$ ($c_{s,v_i^l}$) & SC allocation (cost) for each $g(s,v_i^l)$\\
\hline
$I_a(g(s,v_i))$ & marginal ratio of $g(s,v_i)$\\
\hline
$m$ ($M$) & single (set) maneuver operation\\
\hline
 $\Delta_{v_j}(k)$ & retrieve $k$ SCs from $v_j$\\
\hline
$I_d(\Delta_{v_j}(k))$ & ratio of benefit loss to retrieved SC cost\\
\hline
$\beta_{g(s,v_i)}^{m, M^*}$ ($B$,$C$)& the maneuver (benefit,cost) gap\\
\hline
$\delta K$ & difference of total allocated SCs number\\
\hline
$b^M_j$ ($c^M_j$) & denote the benefit (cost) gain based on $M$\\
\hline
$B^{m,M}_{g(s,v_i)}$ ($C^{m,M}_{g(s,v_i)}$) & benefit (cost) gap after $m$ maneuver operation\\ 
\hline
\textbf{Notation in Proof} & \textbf{Definition} \\
\hline
$V_b^k$ & top $k$ highest influence users in $V_b$ \\
\hline
$c_{ref}$, $C_{opt}$ & reference cost, optimal cost set \\
\hline
$C_1$ ($C_2$) & the first (second) largest cost in $C_{opt}$\\
\hline
$C{sort}$ & sorted cost set from the smallest to the largest\\
\hline
$e_{ref}$ ($e_c$)  &  reference (corresponding candidate) edge set \\
\hline
$b_0$ ($c_0$) & ratio of benefit (cost) variation of users. \\
\hline
$e^*_t$ & maximum benefit edge of $OPT$\\
\hline
$O(M)$ & time of evaluating the expected benefit \\
\hline 
\end{tabular}
\label{Terminology}
\end{table}

\vskip 0.05in
\noindent {\bf {\problem} problem:}

Given an OSN, a propagation model, and the benefit, SC cost, and seed cost for each user, {\problem} is to identify of a seed set $S$, the internal node set $I$, and the SC allocation $K(I)$, to maximize the redemption rate, which is the ratio of the benefit from activated users to the sum of seed costs and SC costs, under an investment budget constraint $B_{inv}$. {\problem} is formulated as follows.

\begin{subequations}
    \begin{gather}
       \mathop{\arg\max}_{S \subset V, I \subset V, K(I)} \frac{B(S,K(I))}{C_{seed}(S)+C_{sc}(K(I))}\label{eq:objevtive},
    \end{gather}
    \text{subject to}
    \begin{gather}
        C_{seed}(S) + C_{sc}(K(I)) \leq B_{inv}.\label{eq:constraint}
    \end{gather}
\end{subequations}
\textblue{The redemption rate has been widely adopted as an effective measure for coupons in marketing research [25], [26], because the redemption rate considers not only the total benefit but also the marginal benefit of the budget (e.g., the additional benefit to spend the last 10\% budget).
Considering a simple example $G=(V,E)$ with two vertices, where $V=\{u,v\}$, $E=\emptyset$, and $B_{inv}=1$. For the two users, coupon costs are $c_{seed}(u)=\epsilon$ and $c_{seed}(v)=1-\epsilon$, and the benefit for them are $1-\epsilon$ and $\epsilon$, respectively, where $\epsilon$ is an arbitrarily small constant. The maximum benefit is $1$ by selecting both vertices, but the maximum redemption rate is $\frac{1-\epsilon}{\epsilon}$ by choosing only $u$ as the seed. User $v$ is not selected since the additional benefit $\epsilon$ generated by choosing $v$ is tiny, but a large budget (i.e., $1-\epsilon$) is required. Therefore, the redemption rate is widely adopted by marketing research to achieve the best performance/cost ratio.} 

\vskip 0.05in
\noindent {\bf Special cases:}
We first give two commonly used coupon strategies, and then we show the strategies and IM are special cases of {\problem} as follows.
1) {\em Limited coupon strategy} is applied by Dropbox, Airbnb, Booking.com, etc., where $k_i$ of each user is an identical constant number, e.g., $k_i=32$ with Dropbox. Thus, it is a special case of {\problem} with a predetermined SC allocation $\hat{K}(I)$, which selects only seeds under the remaining budget (i.e., $B_{inv} - C_{sc}(\hat{K}(I))$). 2) {\em Unlimited coupon strategy} allocates unlimited SCs to users, which implies that the SC cost is marginal (i.e., $c_{sc}(v_i)=0,~\forall v_i \in V$), and is adopted by Uber\footnote{\url{https://help.uber.com/h/27ecd6af-4929-4c53-a81c-f9fbf2432fd4}}, Lyft\footnote{\url{http://get.lyft.com/invites}}, and Hotel.com\footnote{\url{https://refer.hotels.com/friends/us?traffic_source=mWebHP}}. Thus, it is a special case of {\problem}, which is $\mathop{\arg\max}_{S \subset V} \frac{B(S)}{C_{seed}(S)}$, such that $C_{seed}(S) \leq B_{inv}$. Since the SC constraint is removed, the propagation model reduces to the IC model. Similarly, IM is a special case of {\problem} with the SC cost equals to 0 and the SC constraint is unlimited for each user (detail later).

\vskip 0.05in
\noindent {\bf Comparison example:}
Fig. \ref{fig:example} compares IM, PM, and {\problem} with the default parameters (not presented in the figure) as follows. For each user $v_i$, both $c_{seed}(v_i)$ and $c_{sc}(v_i)$ are $1$, and $b(v_i)=3$. The investment budget $B_{inv}$ is $3.5$. Moreover, for the results, the users allocated with SCs and the users within the influence spread are marked in yellow and green, respectively. Since IM and PM do not consider SC allocation, unlimited coupon strategy is applied. For simplicity, we restrict the example of selecting only one seed and allocating two SCs.

Fig. \ref{fig:example}(a) shows the result of IM applying unlimited coupon strategy. For selecting $v_3$ as a seed (i.e., $k_3 = |N(v_3)| = 2$), the expected benefit is $3+3*0.7+3*0.5=6.6$ and the total cost is $C_{seed}(v_3)+C_{sc}(k_3) = 1.5 + (1*0.7+1*0.5) = 2.7$. For $v_1$ and $v_2$ as a seed, the expected benefits are $6.15$ and $4.68$, and the total costs are $2.05$ and $2.1$, respectively. Thus IM allocates two SCs to the selected seed $v_3$ with the maximum expected benefit (i.e., $6.6$) and the redemption rate is $6.6/2.7=2.45$. Fig. \ref{fig:example}(b) is the result of PM applying unlimited coupon strategy. For $v_1$ as a seed, the profit is derived by the expected benefit minus the seed cost (i.e., $(3+3*0.55+3*0.5)-1=5.15$). For $v_2$ and $v_3$ as a seed, the profit are $3.68$ and $5.1$, respectively, and the total costs are the same as IM. Thus, PM allocates two SCs to the selected seed $v_1$, which results in the maximum profit, and the redemption rate is $6.15/2.05=3$.

Fig. \ref{fig:example}(c) shows the result of {\problem}. For selecting $v_1$ as a seed, the remaining budget (i.e., $B_{inv}-c_{seed}(v_1)=2.5$) is invested with three cases as follows: 1) Allocating two SCs to $v_1$. The expected benefit and the total cost are $3+3*0.55+3*0.5=6.15$ and $1+(1*0.55+1*0.5)= 2.05$, respectively, and the redemption rate is $6.15/2.05=3$. 2) Allocating one SC to each of $v_1$ and $v_2$. The expected benefit and the total cost are $5.46$ and $1.975$, respectively. Note that, since $k_1=1$, the probability of activating $v_2$ becomes $(1-0.55)*0.5$ (i.e., $v_4$ fails to redeem the SC) and we call $e(1,2)$ a dependent edge (marked in red) while others independent edge. Thus, the redemption rate is $5.46/1.98=2.76$. 3) Allocating one SC to each of $v_1$ and $v_4$. The expected benefit is $8.295$ and the total cost is $2.675$, and thus the redemption rate is $8.295/2.675=3.1$. The allocation process iteratively calculates the redemption rate by assuming users as seeds. The final result is selecting seed $v_1$ and allocating SCs by $\{k_1=1, k_4=1\}$ with the maximum redemption rate $3.1$.

Fig. \ref{fig:example} shows that {\problem} derives the best redemption rate by carefully examining all possible combinations of seeds and SC allocation. Note that though $c_{seed}(v_4)=c_{seed}(v_5)>B_{inv}$ (i.e., $v_4$ and $v_5$ never become a seed), {\problem} can reap the highest benefit among users (i.e., $b(v_5)$) through allocating an SC to $v_4$ while IM and PM fail to activate $v_5$.

\begin{figure}[t]
    \centering
    \includegraphics[height=1.4in,width=3.5in]{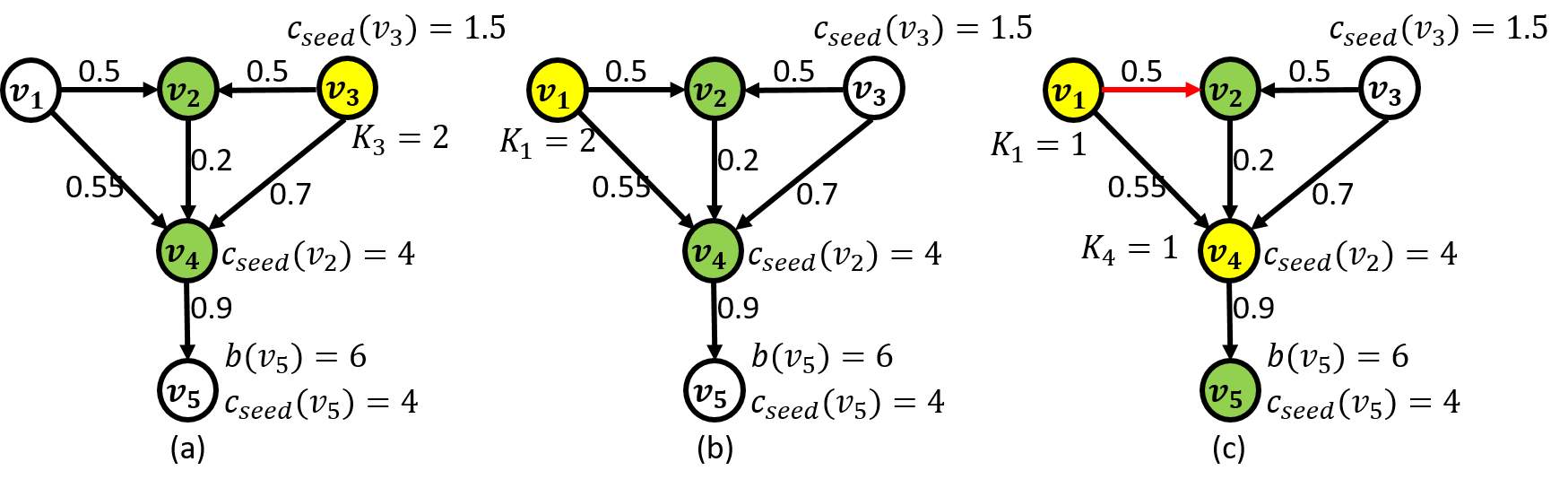}
    \vspace{-4mm} 
    \caption{A simple example comparing three different methods. (a) Result of IM with unlimited coupon strategy. (b) Result of PM with unlimited coupon strategy. (c) Result of {\problem}.}
    \label{fig:example}
\end{figure}

\begin{Theorem}
{\problem} is NP-hard and can not be approximated within $1-1/e+\varepsilon$, where $\varepsilon$ is an arbitrarily small constant.
\end{Theorem}


\begin{proof}

We first give a special case of {\problem}, and then prove that it can be reduced to IM, which derives the hardness of {\problem}. The special case is given as follows. For the OSN $G=\{V,E\}$, the set of users $V=\{v_u \cup V_a \cup V_b\}$ includes an unique user $v_u$ and two identical set of users $V_a$ and $V_b$. For each user $v_i^b \in V_b$, $v_i^b$ connects to only the counterpart $v_i^a \in V_a$ with the weight of $1$. Let $V_b^k$ denote the set of users in $V_b$, who have the top $k$ highest influence. The unique user $v_u$ connects to each user in $V_b^k$ with a weight of $1$. The seed cost of $v_u$ is $k$ while the seed costs of users in $V_a$ and $V_b$ are arbitrarily large, and thus $v_u$ is the only candidate of seed. For each user $v_i^b \in V_b$, let $c_{sc}(v_i^b)=\varepsilon$, where $\varepsilon$ is an arbitrarily small constant. For each user $v_i^a \in V_a$, let $c_{sc}(v_i^a)=0$ which implies that if the counterpart $v_i^b$ is activated by $v_u$, it is activated simultaneously, and then the influence spread of users in $V_a$ is unlimited by the SC constraint (reduces to the IC model). Let $b(v_u) = \varepsilon$, and for each user $v_i^b \in V_b$ and $v_b^a \in V_a$, let $b(v_i^b)=0$ and $b(v_i^a)=1$, respectively. Moreover, let $B_{inv} = k+k\varepsilon$, which implies that the seed $v_u$ can be allocated with at most $k$ SCs.

For the solution of the special case, the seed set $\hat{S}=\{v_u\}$ and the internal node set $\hat{I}=\{v_u \cup V_a \cup V_b^k\}$, where $v_u$ is allocated with $k$ SCs (i.e., $k_u=k$) while the users in $V_a$ have no SC constraints. Thus, the expected benefit is not less than $\varepsilon$ (i.e., benefit of the seed $b(v_u)$) and the total cost ranges from $k$ to $k+k\varepsilon$ (i.e., from seed cost of $v_u$ to $B_{inv}$). Furthermore, the special case can be reduced to IM as follows. Given a IM with the OSN $G'=\{V',E'\}$, where $V'=\{ V_a \cup V_b^k \}$ and $E'$ is the set of edges between users in $V'$, and the seed size of $k$. For any feasible solution of the special case (i.e., $\hat{S}$ and $\hat{I}$), $V_b^k$ is a feasible solution of the IM problem. On the other hand, let $\hat{V}$ denote the seed set selected by solving the IM problem. For any feasible solution of the IM problem, $\hat{V} = V_b^k$ since $v_u$ connects the top $k$ influential users, and thus we can derive a feasible solution of the special case.

For the objective function, the expected benefit derived from the IM problem is slightly less than the expected benefit from the special case by $\varepsilon$, which is negligible. Likewise, the total cost of the special case is approximately equal to $k$ (i.e., $k+k\varepsilon \approx k$) since the SC cost of users in $V_b$ is negligible. Therefore, the special case can be reduced to IM. Furthermore, IM is as hard as maximum $k$--cover problem by the reduction that setting all influence probability to $1$ of IM. Hence, {\problem} is as hard as the maximum $k$--cover problem and the hardness is $1-1/e+\varepsilon$ \cite{Feige}.

\begin{figure}[t]
    \centering
    \includegraphics[height=1.8in,width=3in]{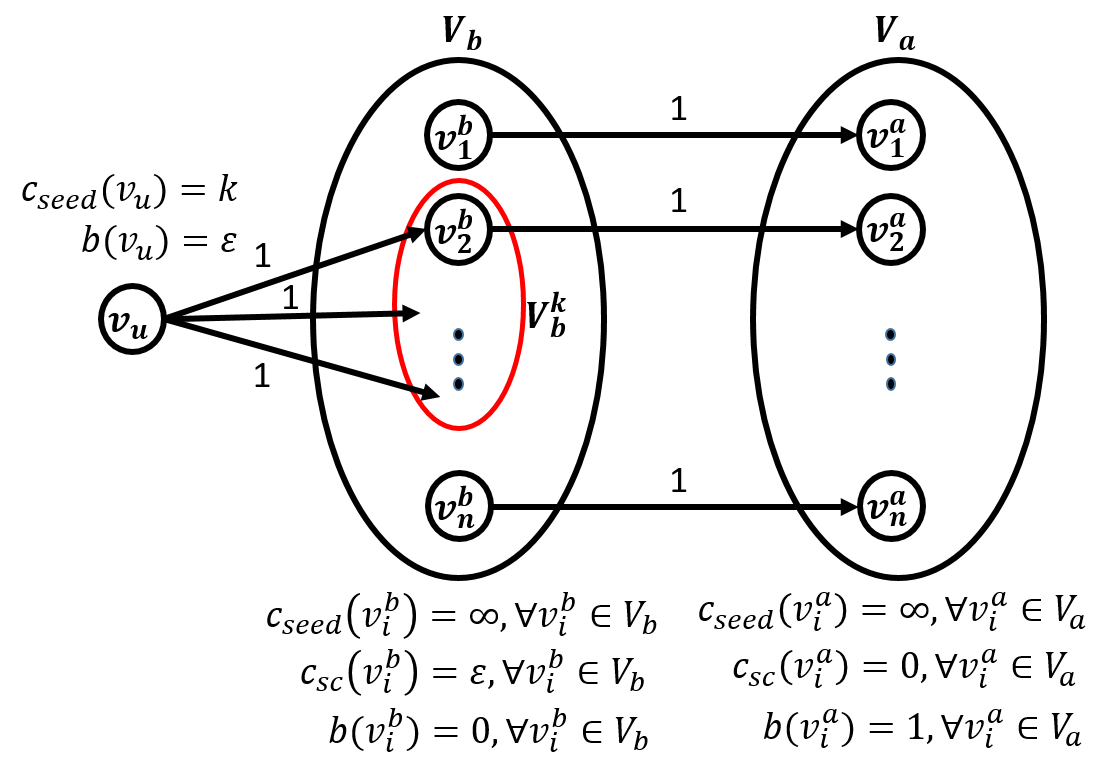}
    \caption{Reduction of {\problem} to IM}
    \label{fig:nphard}
\end{figure}

\end{proof}

\section{Algorithm Design}
\label{sec:algorithm}
To effectively solve {\problem}, we design an approximation algorithm, namely Seed Selection and SC allocation Algorithm ({\algo}). In contrast to IM \cite{Kempe2003Maximizing,chen2010scalable,chen2009efficient,goyal2011data,jung2012irie,leskovec2007cost,nguyen2016cost,nguyen2016stop,ohsaka2014fast,song2015influence,tang2018online,tang2017influence,tang2015influence,tang2014influence,zhou2013ublf} and PM \cite{tang2017profit,tang2018towards,zhu2018host,zhang2016profit,khan2016revenue,lu2012profit,zhu2013influence} focusing on solely seed selection and ignoring the SC allocation, {\algo} introduces the notion of {\em Marginal Redemption (MR)} to prioritize the investment in seeds and allocating SCs to users. For each user, MR represents the ratio of the expected benefit gain to the expected cost gain after activated as a seed or allocated with an extra SC, and therefore a user with a larger MR receives the investment (seed or SC) first. To address the investment trade-off, {\algo} {\em deploys} the investment by carefully examining the MR for three strategies as follows: 1) broadening the current influence spread, 2) deepening the current influence spread, and 3) initiating a new source (seed) of the influence spread.

To address the SC allocation trade-off, {\algo} introduces the notion of {\em Guaranteed Path (GP)} to identify the users with high benefit but receiving no resources (not activated) from the previous deployment. For each user in a GP, the SCs are always sufficient for distribution to ensure the highest activation probability (i.e., no dependent edge) and improve the redemption rate. To evaluate the possible improvement for each GP, {\algo} introduces {\em Amelioration Index (AI)} and {\em Deterioration Index (DI)}, which are the ratio of improved and decreased expected benefit to the allocated and retrieved SC cost. Then, {\algo} can {\em maneuver} the investment in SCs to optimize the redemption rate by examining AI and DI.

{\algo} includes three phases: 1) Investment Deployment (ID), 2) Guaranteed Paths Identification (GPI), and 3) SC Maneuver (SCM). For the investment trade-off, ID exploits MR to iteratively deploy the investment by adopting the strategy with the highest MR. GPI continues to identify the GPs of inactive users with high benefit. Based on the derived GPs, SCM carefully examines the opportunity to maneuver the SCs allocated to users in the ID phase to create a new spread, such that the redemption rate is optimized.

\subsection{Algorithm Description}
\subsubsection{Investment Deployment (ID)}
ID first deploys the investment in seeds and internal nodes under the investment budget by three strategies as follows: 1) activating the source (seed) of a new influence spread, 2) allocating an SC to a user to broaden the current influence spread; meanwhile, enhancing the influence probability by changing an edge from dependent to independent (e.g., by allocating an SC to $v_1$ in Fig. \ref{fig:example}(c), the influence probability improves 27.5\%), and 3) allocating an SC to the user in the end of the current influence spread to deepen the spread. Nonetheless, previous works (IM \cite{Kempe2003Maximizing,chen2010scalable,chen2009efficient,goyal2011data,jung2012irie,leskovec2007cost,nguyen2016cost,nguyen2016stop,ohsaka2014fast,song2015influence,tang2018online,tang2017influence,tang2015influence,tang2014influence,zhou2013ublf} and PM \cite{tang2017profit,tang2018towards,zhu2018host,zhang2016profit,khan2016revenue,lu2012profit,zhu2013influence}) includes no such subtle manipulations of influence spreads.

{\algo} introduces {\em pivot source} for the evaluation of investment in seeds and SCs, which is derived as follows. For each user, ID calculates the MR (detailed later) of the nodes selected as a seed or allocated with an SC after becoming a seed (i.e., updating $K_i$ to $1$), and it then iteratively applies the one with the maximum positive MR. Once a user is selected, it is included in a priority queue $Q$ (sorted by redemption rate). Moreover, when an SC is allocated to a seed in $Q$, the redemption rate is updated. The process stops when all MR becomes negative or all users are included in $Q$. Afterward, $v_i$ is popped from $Q$ as the initial of the influence spread, and it is then included to the candidate seed set $\hat{S}$ and the candidate internal node set $\hat{I}$ (if $K_i=1$). Moreover, \{$\hat{S}$, $\hat{I}$, $K(\hat{I})$\} is included in the candidate deployment $D$. ID further pops $v_p$ from $Q$ as the {\em pivot source} for the comparison of the MR of allocating an SC to users.

For the investment trade-off, ID iteratively examines the MR of allocating an SC to users (strategy 2 and 3) and the redemption rate of the pivot source $v_p$ (strategy 1). If the user $v_i$ has the maximum MR larger than the redemption rate of $v_p$, ID invests an SC in $v_i$ (i.e., increases $K_i$ by 1) and includes $v_i$ in $\hat{I}$ if $v_i \notin \hat{I}$; otherwise, ID initiates a new source $v_p$ and includes $v_p$ in $\hat{S}$ and $\hat{I}$, and it pops the next pivot source from $Q$. In each iteration, $D$ includes the current deployment \{$\hat{S}$, $\hat{I}$, $K(\hat{I})$\}. The process stops when the total cost exceeds the investment budget $B_{inv}$, and the final deployment of ID is $D^* \in D$ with the highest redemption rate. Specifically, let $\gamma_i$ denote the binary variable (initialized as 1) indicating that if a user $v_i$ is selected as seed for the first time, $\gamma_i=1$; otherwise, $\gamma_i=0$. If $\gamma_i=1$, MR of selecting $v_i$ as a seed is defined as $\tau_{\hat{S},\hat{I}}^{v_i,\gamma_i} = \frac{B(\hat{S} \cup v_i, K(\hat{I}))-B(\hat{S},K(\hat{I}))}{C_{seed}(\hat{S} \cup v_i)-C_{seed}(\hat{S})}$; otherwise, MR of allocating an SC to $v_i$ is defined as $\tau_{\hat{S},\hat{I}}^{v_i,\gamma_i} = \frac{B(\hat{S}, K(\hat{I} \cup v_i))-B(\hat{S},K(\hat{I}))}{C_{sc}(K(\hat{I} \cup v_i))-C_{sc}(K(\hat{I}))}$. Note that when $v_i \in \hat{I}$, $K(\hat{I} \cup v_i)$ is different from $K(\hat{I})$ since $K_i$ is increased by 1.

\begin{figure}[t]
    \centering
    \includegraphics[height=1.2in,width=3.5in]{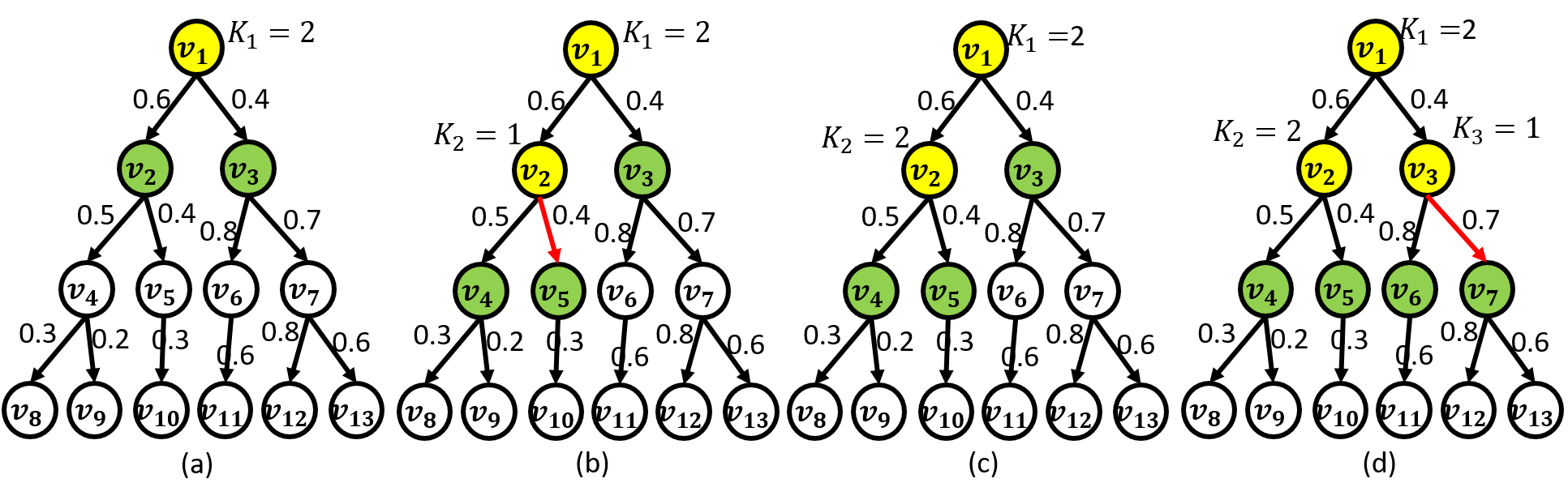}
    \vspace{-8mm} 
    \caption{Results of the example of ID. (a) Iteration 1. (b) Iteration 2. (c) Iteration 3. (d) Iteration 4.}
    \label{fig:RD}
\end{figure}

\begin{example}
Fig. \ref{fig:RD} presents an example of ID with the input setting as follows. For each user $v_i$, $b(v_i)=c_{sc}(v_i)=1$, and $c_{seed}(v_i) \approx \infty$ except that $c_{seed}(v_1) \approx 0$. The dependent edges are marked in red. Since the only possible seed is $v_1$, the queue $Q$ contains only $v_1$, and the influence starts from $v_1$ allocated with an SC. The expected benefit and cost are approximately $1+(1*0.6+1*(1-0.6)*0.4)=1.76$ and $0+1*0.6+1*(1-0.6)*0.4=0.76$, respectively. In the first iteration, the expected benefit gains of allocating an SC to $v_1$ ($K_1=2$), $v_2$ ($K_1=1$, $K_2=1$), and $v_3$ ($K_1=1$, $K_3=1$) are $(1+1*0.6+1*0.4)-1.76=0.24$, $(1+1*0.6+1*(1-0.6)*0.4+1*0.6*0.5+1*0.6*(1-0.5)*0.4)-1.76)=0.42$, $1+1*0.6+1*(1-0.6)*0.4+1*(1-0.6)*0.4*0.8+1*(1-0.6)*0.4*(1-0.8)*0.7)-1.76=0.15$, respectively. Moreover, the expected cost gains of allocating an SC to $v_1$ ($K_1=2$), $v_2$ ($K_1=1$, $K_2=1$), and $v_3$ ($K_1=1$, $K_3=1$) are $(1*0.6+1*0.4)-0.76=0.24$, $(1*0.6+1*(1-0.6)*0.4+1*0.5+1*(1-0.5)*0.4)-0.76=0.7$, and $(1*0.6+1*(1-0.6)*0.4+1*0.8+1*(1-0.8)*0.7)-0.76=0.94$, respectively. Thus, MR of allocating an SC to $v_1$ ($K_1=2$), $v_2$ ($K_1=1$, $K_2=1$), and $v_3$ ($K_1=1$, $K_3=1$) are $0.24/0.24=1$, $0.42/0.7=0.6$, and $0.15/0.94=0.16$, respectively, and the SC is allocated to $v_1$ with the maximum MR. Fig. \ref{fig:RD}(a) shows the result of the first iteration and the result of the following iterations are shown in Fig. \ref{fig:RD}.

Fig. \ref{fig:RD} presents the subtle manipulation of the influence spread by {\algo} with the investment in seeds and SCs, which cannot be accomplished by previous works \cite{Kempe2003Maximizing,chen2010scalable,chen2009efficient,goyal2011data,jung2012irie,leskovec2007cost,nguyen2016cost,nguyen2016stop,ohsaka2014fast,song2015influence,tang2018online,tang2017influence,tang2015influence,tang2014influence,zhou2013ublf,tang2017profit,tang2018towards,zhu2018host,zhang2016profit,khan2016revenue,lu2012profit,zhu2013influence}. IM \cite{Kempe2003Maximizing,chen2010scalable,chen2009efficient,goyal2011data,jung2012irie,leskovec2007cost,nguyen2016cost,nguyen2016stop,ohsaka2014fast,song2015influence,tang2018online,tang2017influence,tang2015influence,tang2014influence,zhou2013ublf} selects the most influential user $v_3$ as the seed. However, the seed cost $c_{seed}(v_3)$ is not considered and can be arbitrarily large. Although PM \cite{tang2017profit,tang2018towards,zhu2018host,zhang2016profit,khan2016revenue,lu2012profit,zhu2013influence} takes the seed cost into consideration, it cannot decide the user (e.g., $v_2$ or $v_3$ in Fig. \ref{fig:RD}(a)) to allocate SCs (even applying existing coupon strategies), and it may result in a miserable redemption rate.

\end{example}

\subsubsection{Guaranteed Paths Identification (GPI)}
For the SC allocation trade-off, GPI constructs the {\em Guaranteed Paths (GP)} to identify the inactive but possibly influenced users based on the result of ID $D^* = \{S^*,I^*,K(I^*)\}$. For each seed $s \in S^*$, the possibly influenced inactive users are restrained by the remaining budget (i.e., subtracting the investment budget by the seed cost). Moreover, for the GP to a possibly influenced inactive user, all edges are independent edges to guarantee the highest influence probability. Nonetheless, existing strategies like unlimited and limited coupon strategies cannot identify these users, and it thereby tends to influence those with high benefits.

For each seed $s \in S^*$, GPI constructs the GP $g(s,v_i)$ from $s$ to each possibly influenced inactive user $v_i$ as follows. GPI first marks $s$ as visited, and it then traverses the descendants of $s$ in DFS manner. For each visited node, the traversal starts from its child with the highest to the lowest influence probability. When a user $v_i$ is visited, GPI includes the following nodes in a temporary set $\hat{g}$, including the visited siblings of both $v_i$ and $v_i$'s ascendants. It then sets $\hat{c}$ as the total expected SC cost of all users in $\hat{g}$. If $\hat{c} \leq B_{inv} - c_{seed}(s)$, 1) $g(s,v_i)$ is set to $\hat{g}$ with a guaranteed cost $c_{s,v_i} = \hat{c}$, 2) for each user $v_j$ in $\hat{g}$, $K_j$ is set to the number of visited children (i.e., each user could receive an SC), 3) it updates the expected benefit $b_{s,v_i}$, and 4) GPI proceeds to traverse $v_i$'s children. Otherwise, GPI stops traversing $v_i$'s children and unvisited siblings (i.e., siblings with influence probability lower than $v_i$'s), and it traverses back to the next sibling of $v_i$'s parent (i.e., the sibling next to $v_i$'s parent in the descending order of influence probability). The process stops when no more user can be visited (i.e., $\hat{c} > B_{inv} - c_{seed}(s)$).

More specifically, for each user $v_i^l$ in the $l$-th level, let $U_s^{{\hat{l}}}$ denote the set of visited siblings in the ${\hat{l}}$-th level before visiting $v_i^l$ by the traversal starting from a seed $s$. A guaranteed path from $s$ to $v_i^l$ is $g(s,v_i^l)=\{v_j | \forall v_j \in  U_s^{{\hat{l}}},{\hat{l}} \leq l\}$. For each $g(s,v_i^l)$, let $\hat{K}(I)=\{\hat{K}_j | \forall v_j \in g(s,v_i^l) \}$ denote its SC allocation, where $\hat{K}_j$ is the SC constraint of $v_j$ in $g(s,v_i^l)$. Moreover, $\hat{K}_j$ equals to the number of $v_j$'s visited children, and the guaranteed cost is $c_{s,v_i^l}=C_{sc}(\hat{K}(I))$ (i.e., the expected SC cost that each user in $g(s,v_i^l)$ could receive an SC).


\begin{figure}[t]
    \centering
    \includegraphics[height=1.2in,width=3.5in]{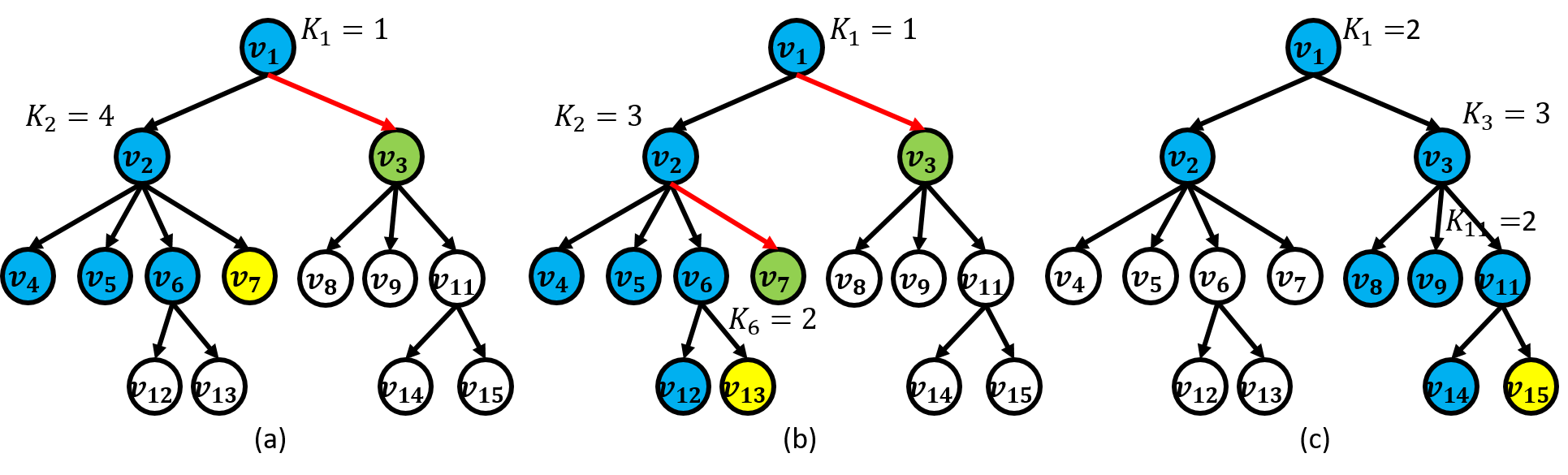}
    \vspace{-8mm} 
    \caption{Examples of guaranteed paths. (a) $g(v_1,v_7)$. (b) $g(v_1,v_{13})$. (c) $g(v_1,v_{15})$}.
    \label{fig:GP}
\end{figure}

\begin{example}
Fig. \ref{fig:GP} presents three examples of GP, and the users of the GPs are marked in blue and yellow (end of the path). For each node, the children on the left side has a higher influence probability. In Fig. \ref{fig:GP}(a), $g(v_1,v_7)=\{v_1,v_2,v_4,v_5,v_6,v_7\}$, and Fig. \ref{fig:GP}(b) and \ref{fig:GP}(c) present $g(v_1,v_{13})$ and $g(v_1,v_{15})$, respectively. Moreover, the expected benefit of a GP involves not only the visited users but also the users connected by the dependent edges (marked in green).
GP is a novel structure specially designed for the SC environment, which is not included in the previous works \cite{Kempe2003Maximizing,chen2010scalable,chen2009efficient,goyal2011data,jung2012irie,leskovec2007cost,nguyen2016cost,nguyen2016stop,ohsaka2014fast,song2015influence,tang2018online,tang2017influence,tang2015influence,tang2014influence,zhou2013ublf,tang2017profit,tang2018towards,zhu2018host,zhang2016profit,khan2016revenue,lu2012profit,zhu2013influence}. GP can identify the possibly influenced inactive users with high benefits and then guide the influence spread for reaping the benefits. However, existing coupon strategies (limited and unlimited) cannot obtain the benefits since they either apply over-simplified methods or ignore the SC allocation.
\end{example}

\subsubsection{SC Maneuver (SCM)}
To optimize the redemption rate, SCM examines the opportunity to maneuver the allocated SCs (i.e., $K(I^*)$) for creating new spreads based on the derived GPs $\mathcal{P}$. SCM introduces {\em Amelioration Index (AI)} and {\em Deterioration Index (DI)} to evaluate the amelioration and the deterioration of redemption rate by allocating and retrieving SCs, respectively. For each $g(s,v_i) \in \mathcal{P}$, AI $I_a(g(s,v_i))$ is the ratio of the expected benefit gains $B_a(g(s,v_i))$ to the expected SC cost of the allocated SCs $C_a(g(s,v_i))$. For retrieving $k$ SCs from $v_j$, DI $I_d(\Delta_{v_j}(k))$ is the ratio of the expected benefit loss to the expected SC cost of the retrieved SCs. Let $m$ and $M$ denote a maneuver operation and a set of maneuver operations (detailed later). To evaluate the efficiency of maneuvering SCs, SCM introduces {\em Maneuver Gap (MG)} $\beta_{g(s,v_i)}^{m,M}$ which is the ratio of the benefit gap to the cost gap. The benefit gap $B^{m,M}_{g(s,v_i)}$ and the cost gap $C_{g(s,v_i)}^{m,M}$ are the difference between the expected benefit gain and the expected SC cost before and after including $m$ in $M$. Equipped with the above notions, {\algo} can examine the improvement of creating each GP, the debasement of redemption rate by retrieving SCs, and the efficiency of maneuvering SCs, whereas these fine-grained operations are not considered in the previous works \cite{Kempe2003Maximizing,chen2010scalable,chen2009efficient,goyal2011data,jung2012irie,leskovec2007cost,nguyen2016cost,nguyen2016stop,ohsaka2014fast,song2015influence,tang2018online,tang2017influence,tang2015influence,tang2014influence,zhou2013ublf,tang2017profit,tang2018towards,zhu2018host,zhang2016profit,khan2016revenue,lu2012profit,zhu2013influence}.

SCM calculates $I_a(g(s,v_i))$ for each $g(s,v_i) \in \mathcal{P}$. From the largest to the smallest $I_a(g(s,v_i))$, if 1) its guaranteed cost $c_{s,v_i}$ does not exceed the total invested SC cost $C_{sc}(K(I^*))$ and 2) $v_i$ cannot be activated by $D^*$ ($K_p \in K(I^*) = 0$, where $v_p$ is $v_i$'s parent), SCM then determines whether to create $g(s,v_i)$ as follows. Let $v_{i^*}$ and $\delta K$ denote the nearest possibly activated ascendant of $v_i$ by $D^*$ and the difference between the total number of SCs allocated in $g(s,v_i)$ and $g(s,v_{i^*})$, respectively. Let $M^*$ denote the set of maneuver operations, where a maneuver operation $m$ includes a DI $I_d(\Delta_{v_j}(k))$, maneuver mapping $\mathcal{K}=\{K_i^j | v_i, v_j \in V \}$ (e.g., $K_i^j=3$ represents maneuvering 3 SCs from $v_j$ to $v_i$), the index of the last user with SCs maneuvered. If the total number of the maneuvered SCs $\sum_{\forall \mathcal{K} \in M^*, \forall K_i^j \in \mathcal{K}} K_i^j < \delta K$, SCM calculates the DIs of all $v_j \in I^*$ and derives a set of corresponding maneuver operations $M$. From the maneuver operation $m \in M$ with the smallest to the largest DI $I_d(\Delta_{v_j}(k))$, if $I_d(\Delta_{v_j}(k)) < \beta_{g(s,v_i)}^{m, M^*}$ and the redemption rate is improved by $\mathcal{K} \in m$, SCM includes $m$ in $M^*$ and proceeds to find the next maneuver operation. If $I_d(\Delta_{v_j}(k)) \geq \beta_{g(s,v_i)}^{m, M^*}$ or the redemption rate is not improved, SCM skips this GP and proceeds to examine the next GP. When $\sum_{\forall \mathcal{K} \in M^*, \forall K_i^j \in \mathcal{K}} K_i^j = \delta K$, $g(s,v_i)$ is created based on $M^*$. The process ends after all GPs in $\mathcal{P}$ are examined.

SCM evaluates the DIs and derives the set of corresponding maneuver operations as follows. Let $K_j$ and $\hat{K_j}$ denote $v_j$'s SC allocation of $K(I^*)$ and $g(s,v_i)$, respectively. For each $v_j \in I^*$, if it has spare SCs for creating $g(s,v_i)$ (i.e., $K_j > \hat{K}_j$), SCM calculates DIs by retrieving all possible numbers $k$ of SCs ($1 \leq k \leq K_j - \hat{K}_j$), as opposite to the derivation of MR (i.e., retrieving instead of adding SCs). Let $\hat{i^*}$ denote the index of user whom SCM maneuvers SCs to, which starts from the nearest possibly activated ascendant of $v_i$. Let $\mathcal{S}$ denote the sum of SCs already maneuvered to $v_{\hat{i^*}}$ according to the current $M^*$ (i.e., $\sum_{\hat{j} \in \hat{\mathcal{K}}} K_{\hat{i^*}}^{\hat{j}}$, where $\hat{\mathcal{K}}$ is the union of all $\mathcal{K} \in M^*$). SCM builds the maneuver mapping $\mathcal{K}$ by splitting the $k$ SCs to $v_{\hat{i^*}}$ and his descendants starting from $v_{\hat{i^*}}$. First, SCM sets $K^j_{\hat{i^*}} = \hat{K}_{\hat{i^*}} - (K_{\hat{i^*}}+\mathcal{S})$ (the number of SCs required for filling $v_{\hat{i^*}}$'s SC allocation to $\hat{K}_{\hat{i^*}}$) if $k$ is large enough. Otherwise, $K^j_{\hat{i^*}} =k$. Then if the total cost after retrieving $K^j_{\hat{i^*}}$ SCs of $v_j$ to $v_{\hat{i^*}}$ is less than $B_{inv}$, $K^j_{\hat{i^*}}$ is included in $\mathcal{K}$. If the number of remaining SCs $\Delta k > 0$, SCM moves to the next target (i.e., set $\hat{i^*}$ to the index of $v_{\hat{i^*}}$'s child with SCs in $g(s,v_i)$). When the mapping of $k$ SCs are determined (i.e., $\Delta k = 0$), the maneuver operation $\{I_d(\Delta_{v_j}(k)), \mathcal{K}, \hat{i^*}\}$ is included in the candidate maneuver operation set $M$. The process stops when all possible $v_j$ are examined.

Specifically, for each $g(s,v_i)$, let $v_j$ denote $v_i$'s nearest activated ascendant and its AI $I_a(g(s,v_i))=\frac{B_a(g(s,v_i))}{C_a(g(s,v_i))}$, where $B_a(g(s,v_i))=b_{s,v_i} - b_{s,v_j}$ and $C_a(g(s,v_i))=c_{s,v_i} - c_{s,v_j}$. Let $b^M_j$ and $c^M_j$ denote the benefit gain and cost gain by maneuvering SCs to $v_j$ and its descendant according to $M$, respectively. The benefit gap and the cost gap are defined as $B^{m,M}_{g(s,v_i)}=b_j^{m \cup M} - b_j^{M}$ and $C_{g(s,v_i)}^{m,M}=c_j^{m \cup M} - c_j^{M}$.

\begin{figure}[t]
    \centering
    \includegraphics[height=1.7in,width=3.5in]{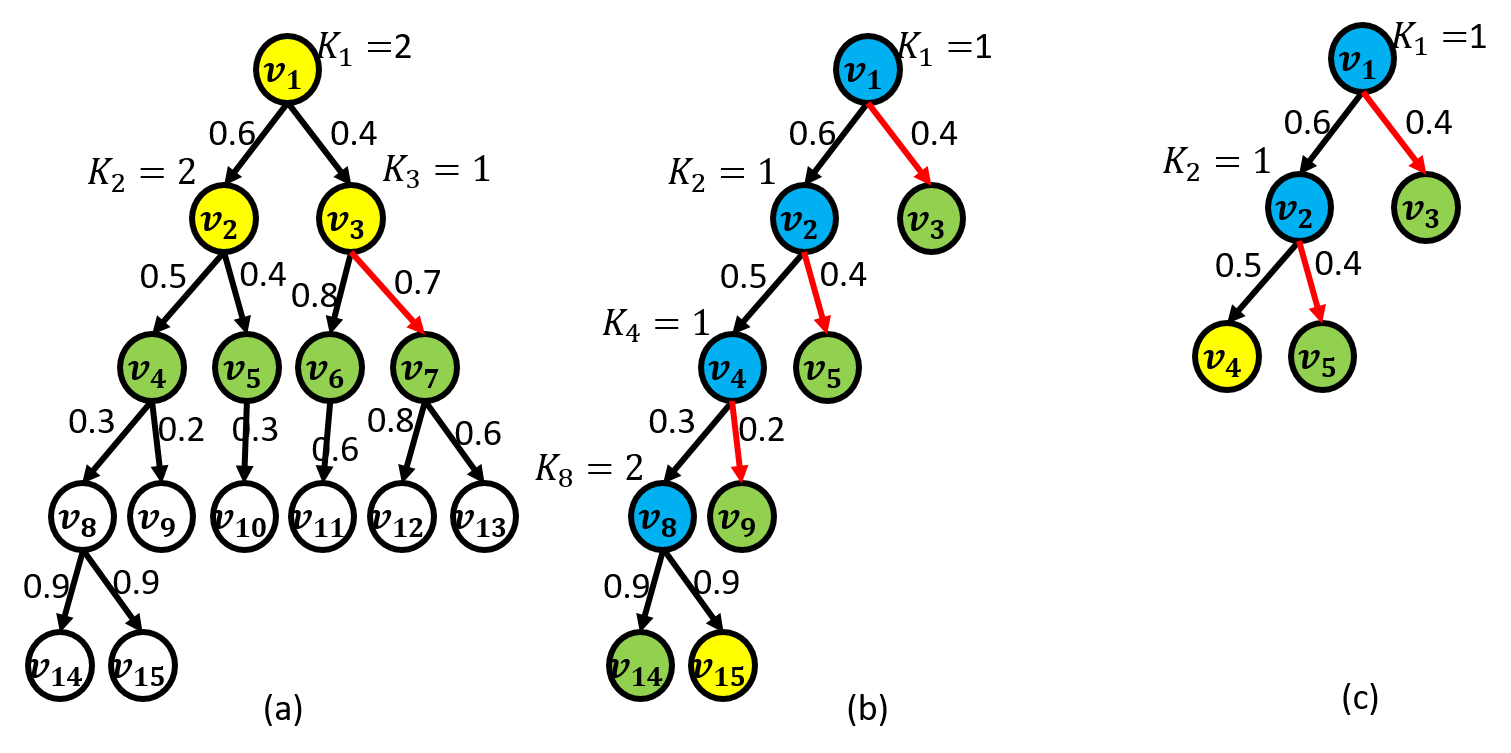}
    \vspace{-3mm} 
    \caption{Example of a maneuver operation. (a) $K_1=2$, $K_2=2$, and $K_3=1$. (b) $g(v_1,v_{15})$. (c) $g(v_1,v_{4})$}
    \label{fig:maneuver}
\end{figure}

\begin{example}
Fig. \ref{fig:maneuver} presents an example of SCM. Fig. \ref{fig:maneuver}(a) is the result of Fig. \ref{fig:RD}(d) attached with two high benefit (50 for each) but inactive users $v_{14}$ and $v_{15}$ with low SC costs. Fig. \ref{fig:maneuver}(b) is $g(v_1,v_{15})$ with $b_{v_1,v_{15}}=10.41$ and $c_{v_1,v_{15}}=2.66$, whereas Fig. \ref{fig:maneuver}(c) presents $g(v_1,v_{4})$ with $b_{v_1,v_{4}}=2.18$ and $c_{v_1,v_{4}}=1.46$. The largest AI is $I_m(g(v_1,v_{15}))=B_m(g(v_1,v_{15})) / C_m(g(v_1,v_{15}))=8.23/1.2=6.86$, and the cost of $g(v_1,v_{15})$ is less than the total cost derived from ID (i.e., $2.66<2.84$). In the first iteration, the sorted DIs are $I_d(\Delta_{v_3}(1))=0.38/0.94=0.4$, $I_d(\Delta_{v_2}(1))=0.12/0.2=0.6$, and $I_d(\Delta_{v_1}(1))=0.47/0.24=1.94$, and thus the only SC of $v_3$ is maneuvered to $v_4$. Then, in the second iteration, the sorted DIs, are $I_d(\Delta_{v_2}(1))=0.12/0.2=0.6$ and $I_d(\Delta_{v_1}(1))=0.24/0.24=1$, and an SC of $v_2$ is maneuvered to $v_8$. Finally, in the last iteration, an SC is maneuvered from $v_1$ to $v_8$, and the redemption rate is improved by $380\%$ (i.e., $3.91/1.03=3.8$).
By carefully examining AI, DI, and MG, {\algo} further optimizes the redemption rate by reaping the high benefits of users guided by GPs. {\algo} shapes the influence spread according to the features of users, which cannot be accomplished by the previous works  \cite{Kempe2003Maximizing,chen2010scalable,chen2009efficient,goyal2011data,jung2012irie,leskovec2007cost,nguyen2016cost,nguyen2016stop,ohsaka2014fast,song2015influence,tang2018online,tang2017influence,tang2015influence,tang2014influence,zhou2013ublf,tang2017profit,tang2018towards,zhu2018host,zhang2016profit,khan2016revenue,lu2012profit,zhu2013influence} applying any existing coupon strategies.

\end{example}
\begin{algorithm}[t]
\algsetup{linenosize=\footnotesize}
\footnotesize
\caption{Seed Selection and SC allocation Algorithm} \label{algo:mainAlgo}
\begin{algorithmic}[1]
    \STATE Initialize a priority queue $Q$, $\gamma_i=1$, $\forall v_i \in \hat{V}$, $\hat{V} \leftarrow V$
    \WHILE{$\hat{V} \neq \emptyset$ \textbf{and} there exists a positive MR}
        \STATE $\hat{\tau} \leftarrow 0$, $\hat{i} \leftarrow \emptyset$
        \FORALL{$v_i \in \hat{V}$}
            \STATE \textbf{if} $\tau_{\hat{S},\hat{I}}^{v_i,\gamma_i} > \hat{\tau}$ \textbf{and} $C_{seed}(v_i) + C_{sc}(\{K_i=1\}) \leq B_{inv}$
            
                \STATE \textbf{then} $\hat{\tau} \leftarrow \tau_{\hat{S},\hat{I}}^{v_i,\gamma_i}$, $\hat{i} \leftarrow i$
        \ENDFOR
        \STATE \textbf{if} $\gamma_{\hat{i}} = 1$ \textbf{then} $\gamma_{\hat{i}} \leftarrow 0$, push $v_i$ to $Q$
        \STATE \textbf{otherwise} $K_i \leftarrow 1$, $\hat{V} \leftarrow \hat{V} \setminus v_{\hat{i}}$
    \ENDWHILE
    
    \STATE $v_i \leftarrow $ pop $Q$, $D \leftarrow \{\hat{S} \leftarrow v_i, \hat{I} \leftarrow v_i, K(\hat{I}) \leftarrow K_i\}$, $v_p \leftarrow $ pop $Q$, $R \leftarrow$ redemption rate of $v_p$ with $K_p$
    \WHILE{$Q \neq \emptyset$ \textbf{and} there exists a positive MR}
    \STATE $\hat{\tau} \leftarrow R$, $\hat{i} \leftarrow p$

        \FORALL{$v_i \in \hat{I}$} 
            \STATE $\Delta K(\hat{I}) \leftarrow $ increase $K_i$ in $K(\hat{I})$ by 1, $\gamma_i \leftarrow 0$
            \STATE \textbf{if} $\tau_{\hat{S},\hat{I}}^{v_i,\gamma_i} > R$ \textbf{and} $C_{seed}(\hat{S}) + C_{sc}(\Delta K(\hat{I})) \leq B_{inv}$
                \STATE \textbf{then} $R \leftarrow \tau_{\hat{S},\hat{I}}^{v_i,\gamma_i}$, $\hat{i} \leftarrow i$
        \ENDFOR

        \FORALL{$v_i \notin \hat{I}$ \textbf{and} $v_i \in V$ is influenced}
            \STATE $\Delta \hat{I} \leftarrow \hat{I} \cup v_i$, $K(\Delta \hat{I}) \leftarrow K(\hat{I}) \cup K_i=1$, $\gamma_i \leftarrow 0$
            \STATE \textbf{if} $\tau_{\hat{S},\Delta \hat{I}}^{v_i,\gamma_i} > R$ \textbf{and} $C_{seed}(\hat{S}) + C_{sc}(K(\Delta \hat{I})) \leq B_{inv}$
                \STATE \textbf{then} $R \leftarrow \tau_{\hat{S},\Delta \hat{I}}^{v_i,\gamma_i}$, $\hat{i} \leftarrow i$
        \ENDFOR
        \STATE \textbf{if} $\hat{i} = p$ \textbf{and} $C_{seed}(\hat{S} \cup v_p) + C_{sc}(K(\hat{I} \cup v_p)) \leq B_{inv}$ 
        \STATE \textbf{then} $\hat{S} \leftarrow \hat{S} \cup \hat{i}$, $\hat{I} \leftarrow \hat{I} \cup \hat{i}$, update $K(\hat{I})$, $v_p \leftarrow$ pop $Q$
        
        \STATE \textbf{else} $\hat{I} \leftarrow \hat{I} \cup \hat{i}$, update $K(\hat{I})$

        \STATE $D \leftarrow D \cup \{\hat{S}, \hat{I}, K(\hat{I})\}$
    \ENDWHILE
    \STATE $D^* \in D$ with the maximum redemption rate, $\mathcal{P} \leftarrow \emptyset$
    
    \STATE $\mathcal{P} \leftarrow$ call GPI procedure (Alg. \ref{algo:GPI}) with input $\{G,D^*,B_{inv}\}$
    
                    
    
    \STATE Sort $I_a(g(s,v_i))$ in descending order: $I^1_a$, $I^2_a$,$\cdots$,$I^{|\mathcal{P}|}_a$; $v_p$ is $v_i$'s parent
    
    
    
    \FOR{$I_a^n(g(s,v_i))$, $n=1$ to $|\mathcal{P}|$}
        \IF{$c_{s,v_i} \leq C_{sc}(K(I^*))$ \textbf{and} $K_p \in K(I^*) = 0$}
        
            \STATE $v_{i^*} \leftarrow$ the nearest possibly activated ascendant of $v_i$ 

            \STATE $\delta K \leftarrow \sum_{v_j \in g(s,v_i)} K_j - \sum_{v_j \in g(s,v_{i^*})} K_j$ 
            \STATE $M^* \leftarrow \{0,\mathcal{K}=\emptyset,i^*\}$
            
            \WHILE{$\sum_{\forall \mathcal{K} \in M^*, \forall K_i^j \in \mathcal{K}} K_i^j < \delta K$}
                \STATE $M \leftarrow$ call DIMD procedure (Alg. \ref{algo:DIMD}) with input \{$G$, $D^*$, $g(s,v_i)$, $M^*$\}
                
                \FOR{$m=\{I_d(\Delta_{v_j}(k)), \mathcal{K}, \hat{i^*}\} \in M$ in ascending order according to $I_d(\Delta_{v_j}(k))$}
                    
                    \STATE \textbf{if} $I_d(\Delta_{v_j}(k)) < \beta_{g(s,v_i)}^{m, M^*}$ \textbf{and} redemption rate increases 
                        \STATE \textbf{then} $M^* \leftarrow M^* \cup m$, $i^* \leftarrow \hat{i^*}$, break for-loop

                    \STATE \textbf{if} All $I_d(\Delta_{v_j}(k)) > \beta_{g(s,v_i)}^{m, M^*}$ \textbf{or} redemption rate decreases
                    
                        \STATE \textbf{then} Skip this GP, proceed to the next GP
                \ENDFOR
            \ENDWHILE
            \STATE Update $D^*$ according to $M^*$
        \ENDIF
    \ENDFOR
\end{algorithmic}
\end{algorithm}

\begin{algorithm}[t]
\algsetup{linenosize=\footnotesize}
\footnotesize
\caption{Guaranteed Paths Identification (GPI)}\label{algo:GPI}
{\bf Input:} $G=\{V,E\}$, $D^*=\{S^*,I^*,K(I^*)\}$, and $B_{inv}$ \\
{\bf Output:} $\mathcal{P}$
\begin{algorithmic}[1]
    \FORALL{$s \in S^*$}
        \STATE Push $s$ in stack $\Psi$, $U^{l_i}_s \leftarrow \emptyset$, $visited[v_i] \leftarrow false$, $\forall i$
        \WHILE{$\Psi \neq \emptyset$}
            \STATE $v_i \leftarrow$ pop $\Psi$
            \IF{$visited[v_i] = false$}
                \STATE $\hat{g} \leftarrow \{v_i\} \cup \{ v_j | \forall v_j \in U^{l_j}_{s}, \forall l_j \leq l_i \}$, $\hat{c} \leftarrow \sum_{v_j \in \hat{g}}c_{sc}(v_j)$
                \IF{$\hat{c} \leq B_{inv} - c_{seed}(s)$}
                    \STATE $visited[v_i] \leftarrow true$, $g(s,v_i) \leftarrow \hat{g}$, $c_{s,v_i} \leftarrow \hat{c}$ 
                    \STATE $U^{l_i}_s \leftarrow U^{l_i}_s \cup v_i$, $b_{s,v_i} \leftarrow \sum_{v_j \in \hat{g}}b(v_j)$
                    \STATE $\mathcal{P} \leftarrow \mathcal{P} \cup \{ g(s,v_i), c_{s,v_i}, b_{s,v_i}\}$, push $v_j \in N(v_i)$ in $\Psi$ ascendant order of $p(i,j)$
                    
                \ENDIF
        \ENDIF
        \ENDWHILE
    \ENDFOR
\end{algorithmic}
\end{algorithm}

\begin{algorithm}[t]
\algsetup{linenosize=\footnotesize}
\footnotesize
\caption{Derivation of DI and Maneuver Operation (DIMD)}\label{algo:DIMD}
{\bf Input:} $G=\{V,E\}$, $D^*=\{S^*,I^*,K(I^*)\}$, $g(s,v_i)$, and $M^*$ \\
{\bf Output:} $M$
\begin{algorithmic}[1]
    \STATE $M \leftarrow \emptyset$, $i^* \leftarrow$ $i^*$ of the last attached element of $M^*$
    \FORALL{$v_j \in I^*$}
        \IF{$K_j > \hat{K}_j$}
            \FOR{$k=1$ to $K_j - \hat{K}_j$}
                \STATE $\Delta K(I^*) \leftarrow$ $K(I^*)$ with $K_j$ decreased by $k$
                \STATE $I_d(\Delta_{v_j}(k)) \leftarrow \frac{B(S^*,K(I^*)) - B(S^*,\Delta K(I^*))}{C_{sc}(S^*,K(I^*)) - C_{sc}(S^*,\Delta K(I^*))}$
                \STATE $\mathcal{K} \leftarrow \emptyset$, $\Delta k \leftarrow k$, $\hat{i^*} \leftarrow i^*$, $\hat{\mathcal{K}} \leftarrow$ union of all maneuver mapping in $M^*$, $\mathcal{S} \leftarrow \sum_{\hat{j} \in \hat{\mathcal{K}}} K_{\hat{i^*}}^{\hat{j}}$
                
                \WHILE{$K_{\hat{i^*}} + \mathcal{S} < \hat{K}_{\hat{i^*}}$ and $\Delta k > 0$}
                    \STATE \textbf{if} $\hat{K}_{\hat{i^*}} - (K_{\hat{i^*}}+\mathcal{S}) < \Delta k$
                        \STATE \textbf{then} $K^j_{\hat{i^*}} \leftarrow \hat{K}_{\hat{i^*}} - (K_{\hat{i^*}}+\mathcal{S})$, $\Delta k \leftarrow \Delta k - K^j_{\hat{i^*}}$
                        
                    \STATE \textbf{if} $\hat{K}_{\hat{i^*}} - (K_{\hat{i^*}}+\mathcal{S}) = \Delta k$
                        \STATE \textbf{then} $K^j_{\hat{i^*}} \leftarrow \Delta k$, $\Delta k \leftarrow 0$
                    \STATE \textbf{if} the total cost after retrieving $K^j_{\hat{i^*}}$ SCs from $v_j$ to $v_{\hat{i^*}}$ is less than $B_{inv}$
                    
                        \STATE \textbf{then} $\mathcal{K} \leftarrow \mathcal{K} \cup K^j_{\hat{i^*}}$. 1) If $\Delta k \neq 0$, $\hat{i^*} \leftarrow $ the descendant of $v_{\hat{i^*}}$ with SCs in $g(s,v_i)$. 2) If $\Delta k = 0$, $M \leftarrow M \cup \{I_d(\Delta_{v_j}(k)), \mathcal{K}, \hat{i^*}\}$
                        \STATE \textbf{Else} Break the while-loop
                \ENDWHILE
            \ENDFOR
        \ENDIF
    \ENDFOR
\end{algorithmic}
\end{algorithm}
\textblue{
Algorithm 1 presents the pseudo-code of {\algo} with three phases: 1) Investment Deployment (ID) (Line 1--24), 2) Guaranteed Paths Identification (GPI) (Line 25--34), and 3) SC Maneuver (SCM) (Line 35--48). Specifically, ID first identifies the {\em pivot source} by iteratively selecting a user $v_i$ with the maximum positive marginal redemption (MR) and adding it to a queue $Q$ prioritized by redemption rate (Line 4--7). 
Then, the initial seed is popped from $Q$ as the initial investment deployment $D$, and the pivot source $v_p$ is popped from $Q$ to compare the MR of users not in $Q$ (Line 9). Subsequently, ID deploys the investment by iteratively examining the MR of each user and the redemption rate of $v_p$ (Line 10--23). Three strategies are applied to broaden and deepen the influence (Line 12--15, Line 16--19) by allocating an SC to $v_i \in \hat{I}$, $v_i \notin \hat{I}$, and then activate a new seed (Line 20--21). The user with the maximum value of MR or redemption rate is invested by an SC or assigned as a seed.

Moreover, the guaranteed paths GPs are identified by the GPI procedure (Alg. \ref{algo:GPI})
in a Depth-First Search manner. On the same level, GPI traverses users in ascending order of their influence probabilities (Line 25). \textredTech{In Alg. \ref{algo:GPI},} GPI traverses from each seed $s$ in $S^*$ to find GPs $g(s,v_i)$ to each $v_i$, where $g(s,v_i)$ is the set of the visited siblings of both $v_i$ and $v_i$'s ascendants \textredTech{(Line 6).}
When GPI visits $v_i$, if $g(s,v_i)$'s cost (total SC cost of users in $g(s,v_i)$) is smaller than the budget, $g(s,v_i)$ is included in a GP set $\mathcal{P}$.
Otherwise, GPI stops traversing $v_i$'s children and unvisited siblings, and it traverses back to the next sibling of $v_i$'s parent. The process stops when no more user can be visited \textredTech{(Line 7--10)}.

Finally, SCM determines whether to create each $g(s,v_i) \in \mathcal{P}$ based on the Amelioration Index (AI) $I_a(g(s,v_i))$, which is the ratio of the expected benefit gain to the expected SC cost (Line 27--39). SCM examines $I_a(g(s,v_i))$ in descending order and checks the Deterioration Index (DI) $I_d(\Delta_{v_j}(k))$ of each influenced user $v_j$, which is the ratio of the expected benefit loss to the expected SC cost by retrieved $\Delta_{v_j}(k)$ SCs, according to the DIMD procedure 
(Line 33). 
\textredTech{Specifically, in the DIMD procedure (Alg. \ref{algo:DIMD}), let $K_j$ and $\hat{K_j}$ denote $v_j$'s SC allocation of $K(I^*)$ and $g(s,v_i)$, respectively. For each $v_j \in I^*$, if it has spare SCs for creating $g(s,v_i)$ (i.e., $K_j > \hat{K}_j$), DIMD calculates DIs by retrieving all possible numbers $k$ of SCs ($1 \leq k \leq K_j - \hat{K}_j$) (Line 2--15). SCM sets $K^j_{\hat{i^*}}$, which is the number of maneuvered SCs from $v_j$ to the target user $v_{\hat{i^*}}$, depending on $\Delta k$ (Line 9--12). Subsequently, if the total cost after retrieving $K^j_{\hat{i^*}}$ SCs of $v_j$ to $v_{\hat{i^*}}$ is less than $B_{inv}$, $K^j_{\hat{i^*}}$ is included in $\mathcal{K}$. Then, if the number of remaining SCs $\Delta k > 0$, SCM moves to the next target. Otherwise $\Delta k = 0$, the maneuver operation $\{I_d(\Delta_{v_j}(k)), \mathcal{K}, \hat{i^*}\}$ is included in the candidate maneuver operation set $M$ (Line 13--14). DIMD stops when all possible $v_j$ are examined and returns to Alg. \ref{algo:mainAlgo}. 
}
DIMD returns all possible influenced users to retrieve SCs for maneuvering SCs to reduce the redemption loss. Besides, a set of the maneuver operations $M$, which records the detail of retrieve/maneuver SCs, are performed. Based on $M$, SCM decides whether to create $g(s,v_i)$ to optimize the redemption rate (Line 34--38). Then, the final investment deployment $D^*$ is updated (Line 39).}
\section{Performance Analysis}
\label{sec:analysis}

In the following, we first prove that the expected benefit $B(S,K(I))$ is a non-submodular function and derive the {\em reference edges} based on Observation \ref{ob:guessing}, which finds the relation between the optimal solution (\sdollar{$OPT$}) and the pivot sources derived from ID in Sec. \ref{sec:algorithm}. We analyze $B(S,K(I))$ to keep a certain level of submodularity. Subsequently, we prove that {\algo} is a \sdollar{$(1-e^{-\frac{1}{c}}-\varepsilon)$}-approximation algorithm, where $c$ is a constant. 


\begin{lemma}\label{lemma:non-submodular}
For the {\problem} problem, the expected benefit function is non-submodular, and the cost functions are modular.
\end{lemma}
\begin{proof}
For ease of description, we transform {\problem} to an edge selection problem as follows. We first add a virtual node $t$ to the OSN $G=\{V,E\}$, which connects to all vertices in $V$. For selecting each seed $s$, it is transformed to select the edge $e(t,s)$ and assign its weight (influence probability) to 1, whereas for each non-seed user, the weight is 0. The propagation begins from $t$, and then $t$ activates the selected seeds with their seed costs. Second, similar to \cite{Kempe2003Maximizing}, we conduct the analysis on an implemented deterministic graph $G'$. For each user $v_i$, the SC allocation (i.e., assigning $K_i$) is transformed to selecting the out-edges (friends) iteratively from the one with the highest influence probability in $G'$. The expected benefit is updated based on the influence propagation of the selected edges and the selection process stops when the total cost exceeds $B_{inv}$. Thus, the problem is transformed to $\mathop{\arg\max}_{E' \subset E} \frac{B(E')}{C(E')}$, where $E'$ is the selected edges.

For each edge $e(i,j)$, we call it a {\em live edge} if $v_i$ successfully activates $v_j$; otherwise, we call it a {\em pseudo-live edge}. For each live edge, there must exist one path from the virtual node, where all edges are live, and thus we call it a {\em live path}. Let $E_1$ and $E_2$ denote any two set of selected edges, where $E_1 \subseteq E_2 \subseteq E$. We prove that the expected benefit function $B(e)$ is non-submodular by including each edge $\hat{e} \in E \setminus E_2$ in both $E_1$ and $E_2$ with three cases as follows. 1) $\hat{e}$ is pseudo-live after being included in both $E_1$ and $E_2$. We have $B(E_2 \cup \hat{e}) - B(E_2) = B(E_1 \cup \hat{e}) - B(E_1) = 0$ since no user becomes active (i.e., no benefit gain). 2) $\hat{e}$ is live after being included in both $E_2$ and $E_1$, which implies that $\hat{e}$ is either attached to a live path existing in both $E_2$ and $E_1$ or the cause to a path turning to live. For the first case, $B(E_2 \cup e) - B(E_2) \leq B(E_1 \cup e) - B(E_1)$, whereas for the latter one, $B(E_2 \cup e) - B(E_2) \geq B(E_1 \cup e) - B(E_1)$. 3) $\hat{e}$ is live after being included in $E_2$ but pseudo-live in $E_1$, which implies that $\hat{e}$ is attached to a live path existing in $E_2$ but not in $E_1$. Thus, $B(E_2 \cup e) - B(E_2) \geq B(E_1 \cup e) - B(E_1)$.


For each selected edge $e(i,j)$, its cost is $c_{seed}(j) P(e(i,j))$ or $c_{sc}(j) P(e(i,j))$, when $v_i$ is the virtual node or user, respectively. Note that, $G'$ includes no dependent edge since $P(e(i,j))$ is determined, and thus the cost of non-selected edges is $0$. The total cost function $C(e)$ is the sum of the cost of the selected edges. Let $E_1$ and $E_2$ denote any two set of selected edges, where $E_1 \subseteq E_2 \subseteq E$. By including each edge $\hat{e} \in E \setminus E_2$ in both $E_1$ and $E_2$, $C(E_2 \cup \hat{e}) - C(E_2) = C(E_1 \cup \hat{e}) - C(E_1) = C(\hat{e})$, and thus the cost function $C$ is modular.
\end{proof}

\begin{observation}\label{ob:guessing}
Given a reference cost $c_{ref}$, any set of costs that includes $c_{ref}$ and a cost higher than $c_{ref}$, or two identical costs of $c_{ref}$, is not the set of costs of the optimal solution $C_{opt}$.
\end{observation}


Let $C_1$ and $C_2$ denote the first and second largest costs in $C_{opt}$. Note that the cost includes the seed cost and the SC cost since we conduct the analysis on the transformed edge selection problem. Moreover, let $C_{sort}$ denote the sorted costs (from the smallest to the largest), which has $|E'|$ elements ($|E'|$ is the number of edges of the transformed problem). For each cost in $C_{sort}$ as the reference cost $c_{ref}$, we assume it is $C_1$ with the following properties: 1) if $C_1 = C_2$, any set contains $c_{ref}$ and a cost higher than $c_{ref}$ is not $C_{opt}$, and 2) if $C_1 > C_2$, any set containing two elements of $c_{ref}$ is not $C_{opt}$.

Based on Observation \ref{ob:guessing}, we can derive the {\em reference edge set} $e_{ref}$ similar to the pivot source as follows. First, we initialize a candidate set of reference edges $\hat{e}=\emptyset$. For each cost in $C_{sort}$ as the reference cost $c_{ref}$, let $C_t = \{c | \forall c \leq c_{ref} \in C_{sort}~\text{or}~\forall c < c_{ref} \in C_{sort}\}$ denote the target set of costs, which corresponds to the assumptions $C_1 = C_2$ (i.e, $\forall c \leq c_{ref}$) and $C_1 > C_2$ (i.e., $\forall c < c_{ref}$). Then the set of edges within $C_t$ with the first two highest MRs are selected iteratively and included in $\hat{e}$. Finally, the edge set in $\hat{e}$ with the highest redemption rate is selected as $e_{ref}$. Note that $e_{ref}$ contains either two edges from the virtual node or one path from the virtual node, which corresponds to selecting two seeds or a seed allocated with an SC in {\problem}, respectively.

\begin{lemma}\label{lemma:merge}
 Within the space $C^*_t$ strictly contains $OPT$, where $c_{ref}=C_1$. if the corresponding candidate edge set $e_c$ has $\frac{B(e_c)}{C(e_c)} \geq C_0 \frac{B(OPT)}{C(OPT)}$, where $C_0 \leq 1$ is a constant, then the result of {\algo} is always no less than $(1-\epsilon)C_0 \frac{B(OPT)}{C(OPT)}$, where $\epsilon$ is an arbitrary constant.
\end{lemma}

\begin{proof}
We adopt $e_{ref}$ as the connection between the result of {\algo}, $e_c$, and OPT. First, the relation between $e_c$ and OPT is proved to find the relation between $e_{ref}$ and OPT. Then, we prove the relation between {\algo} and $e_{ref}$ and subsequently the one between {\algo} and OPT.

First, we prove that $\frac{B(e_{ref})}{C(e_{ref})} \geq C_0 \frac{B(OPT)}{C(OPT)}$ as follows. Because the set of costs of the optimal solution $C_{opt}$ never resides in $C_{sort} \setminus C_t$ for each $c_{ref}$ based on Observation \ref{ob:guessing}, we derive the reference edges $e_{ref}$ from the target cost sets $C_t$ of all reference costs $c_{ref}$, where $C_{opt}$ must exist. When $c_{ref}=C_1$, $C_t^*$ is the smallest space containing $C_{opt}$, and its corresponding candidate edges $e_c$ is derived as the edges with the first two highest MRs are included in $\hat{e}$. Moreover, $e_{ref}$ has the highest redemption rate among all edge sets in $\hat{e}$ (i.e., $\frac{B(e_{ref})}{C(e_{ref})} \geq \frac{B(e_c)}{C(e_c)}$). Thus, if $\frac{B(e_c)}{C(e_c)} \geq C_0 \frac{B(OPT)}{C(OPT)}$, where $C_0 \leq 1$ is a constant, then $\frac{B(e_{ref})}{C(e_{ref})} \geq C_0 \frac{B(OPT)}{C(OPT)}$.

Note that $\frac{B(e_{ref})}{C(e_{ref})}$ can be derived easily since $e_{ref}$ contains at most two edges. However, for {\algo}, the calculation of the expected benefit $B(S,K(I))$ is more difficult due to the influence propagation \cite{Kempe2003Maximizing}. Hence, following previous works, $B(S,K(I))$ can be obtained approximately by sampling methods, such as Monte Carlo \cite{Kempe2003Maximizing} and reverse greedy \cite{tang2014influence}. 
\textblue{More specifically, it first tosses a coin for each edge with the given influence probability to generate a graph. The users reachable from the seed set by the paths with the allocated coupons will be activated. Note that if a user $v_i$ is allocated with more than $k_i$ coupons with the corresponding $k_i$ living edges after tossing coins, it will only receive the former $k_i$ coupons from the incident edges with the largest influence probability.}
The accuracy of estimating $B(S,K(I))$ increases as the number of sampling increases, and an arbitrarily close approximation can be obtained, and thus the estimation is $(1-\epsilon)$--approximated.

Finally, we find the relationship between the result of {\algo} and OPT as follows. For the pivot source identification in Sec. \ref{sec:algorithm}, all users are assumed as a seed and included in the priority queue $Q$. For each user as a seed, the MR after activation and the MR of allocation with an SC after activation are examined. {\algo} invests a seed or an SC with the highest positive MR iteratively, and the user is included in $Q$ (with one SC or not). The first two elements of $Q$ are popped out as the initial influence spread and the pivot source. Let $s_1$ and $s_2$ denote the first two elements of $Q$. Let $e_1^u$ ($e_2^u$) and $e_1^d$ ($e_2^d$) denote the edges from the virtual node and the ones to the activated child by $s_1$ ($s_2$), respectively. The result of {\algo} is at least as large as $(1-\epsilon)\frac{B(\{e_1^u,e_1^d,e_2^u,e_2^d\})}{C(\{e_1^u,e_1^d,e_2^u,e_2^d\})}$, since {\algo} ensures an improvement in each movement (i.e., ID, GPI, and SCM) with an estimated bias of $B(S,K(I))$, i.e., $(1-\epsilon)$.

Let $e_{ref}^1$ and $e_{ref}^2$ denote the two edges of $e_{ref}$, which have the first two highest MRs. Furthermore, $e_{ref}^1$ and $e_{ref}^2$ are iteratively selected in the same way of selecting $\{e_1^u,e_1^d,e_2^u,e_2^d\}$, and thus $e_{ref}^1$ and $e_{ref}^2 \in \{e_1^u,e_1^d,e_2^u,e_2^d\}$. Therefore, $\frac{B(\{e_1^u,e_1^d,e_2^u,e_2^d\})}{C(\{e_1^u,e_1^d,e_2^u,e_2^d\})} \geq \frac{B(e_{ref})}{C(e_{ref})}$, implying that the result of {\algo} is always no less than $(1-\epsilon)C_0\frac{B(OPT)}{C(OPT)}$.

\end{proof}


After analyzing the relations among the result of {\algo}, redemption rate of $e_{ref}$, and the redemption rate of $e_c$, we further analyze the relation between the redemption rate of $OPT$ and $e_c$. Since the expected benefit function is non-submodular according to Lemma \ref{lemma:non-submodular}, we first derive an upper bound of the expected benefit function and then find the approximation ratio of {\algo} as follows.

Let $X_j$ denote the result of $e_c$ in the $j$-th iteration, where $j=0,1,2$, and $b_0 = \frac{\max(b(v_i))}{\min(b(v_j))},~\forall v_i, v_j \in V$, denotes the ratio of the maximum benefit to the minimum benefit of users. Let $e^*_t$ denote the selected edge with the maximum benefit of $OPT$. By assuming $OPT\setminus X_1=e^*_1,e^*_2, \cdots,e^*_k$, we have the following lemma.


\begin{lemma}\label{lemma:decomposition}
There exists a decomposition of $OPT$ in some order such that for $j=0,1$, $\sum^k_{i=1}(B(X_j \cup \{e^*_1 \cup \cdots \cup e^*_{i}\})- B(X_j \cup \{e^*_1 \cup \cdots \cup e^*_{i-1}\})\leq b_0\sum^k_{i=1}B(e^*_{t})$.
\end{lemma}

\begin{proof}
Since the influence spread starts from the virtual node of the transformed problem, it is intuitive that $OPT$ forms a connected graph rooted at the virtual node; otherwise, there is no gain of redemption rate. Similarly, $e_c$ is also a rooted connected graph. Thus, for $OPT$, there exhibits an order of edges becoming live from the root, which is denoted as $Y_i=\{e^*_1 \cup \cdots \cup e^*_{i}\}$, where the activation begins from $1$ to $i$. Let $Y=OPT\setminus X_1=\{e^*_1\cup...e^*_{k}\}$ by the activation order from $1$ to $k$. For any $i=1,\cdots,k, \ B(X_1 \cup \{e^*_1 \cup \cdots \cup e^*_{i}\})- B(X_1 \cup \{e^*_1 \cup \cdots \cup e^*_{i-1} \})\leq B(e^*_{i}) \leq b_0B(e^*_{t})$. Hence, $\sum^k_{i=1}(B(X_1 \cup \{e^*_1 \cup \cdots \cup e^*_{i} \})- B(X_1 \cup \{e^*_1 \cup \cdots \cup e^*_{i-1}\})\leq \sum^k_{i=1}b_0B(e^*_{t})$. For $j=0$ (i.e., $X_0$), we can derive the same result following the same analysis; therefore, the lemma is proved.
\end{proof}

Although the expected benefit function is non-submodular, with the upper bound derived by Lemma \ref{lemma:decomposition}, let $c_0=\frac{\max(c_{seed}(v_i) \cup c_{sc}(v_i))}{\min(c_{seed}(v_j) \cup c_{sc}(v_j))}$, $\forall v_i, v_j \in V$ denote the ratio of the maximum cost to the minimum cost of users, the performance bound of {\algo} is derived as follows.

\begin{Theorem}
{\algo} is a $(1-e^{-\frac{1}{b_0c_0}}-\epsilon)$--approximation algorithm for the {\problem}.
\end{Theorem}

\begin{proof}

Since the expected benefit function is monotone increasing, by the decomposition Lemma \ref{lemma:decomposition}, we have
\par\noindent
$B(OPT)\leq B(X_1\cup (OPT\setminus X_1)) \\
= B(X_1)+\sum^k_{i=1}(B(X_1\cup \{e^*_1 \cup \cdots \cup e^*_{i}\})- B(X_1 \cup \{e^*_1 \cup \cdots \cup e^*_{i-1}\}) \\
\leq B(X_1)+b_0\sum^k_{i=1}e^*_{t} \\
=B(X_1)+b_0\sum^k_{i=1}(B(X_1\cup e^*_{t})-B(X_1)) \\
=B(X_1)+b_0\sum^k_{i=1}\frac{B(X_1\cup e^*_{t})-B(X_1)}{C(X_1\cup e^*_{t})-C(X_1)}(C(X_1\cup e^*_{t})-C(X_1))$.

By iteratively selecting the edges with the maximum MR,
\par\noindent $B(X_1)+b_0\sum^k_{i=1}\frac{B(X_1\cup e^*_{t})-B(X_1)}{C(X_1\cup e^*_{t})-C(X_1)}(C(X_1\cup e^*_{t})-C(X_1))\leq B(X_1)+b_0\sum^k_{i=1}\frac{B(X_1\cup x_{2})-B(X_1)}{C(X_1\cup x_{2})-C(X_1)}(C(X_1\cup e^*_{t})-C(X_1))\leq 
 B(X_1)+b_0\sum^k_{i=1}\frac{B(X_1\cup x_{2})-B(X_1)}{C(X_1\cup x_{2})-C(X_1)}c_0(C(X_1\cup e^*_{i})-C(X_1))\leq B(X_1)+b_0c_0\frac{B(X_2)-B(X_1)}{C(X_2)-C(X_1)}C(OPT)$.

The last inequality holds since $C$ is a linear function. Thus,
\begin{equation}\label{eq:eq1}
B(OPT)-B(X_1)\leq b_0c_0\frac{B(X_2)-B(X_1)}{C(X_2)-C(X_1)}C(OPT).
\end{equation}
By applying the same analysis on $X_0$, 
\begin{equation}\label{eq:eq2}
B(OPT)-B(X_0)\leq b_0c_0\frac{B(X_1)-B(X_0)}{C(X_1)-C(X_0)}C(OPT).
\end{equation}

Let $\delta_i=B(OPT)-B(X_i)$ and $\alpha_i=\frac{b_0c_0C(OPT)}{C(X_i)-C(X_{i-1})}$. By Inequalities \ref{eq:eq1} and \ref{eq:eq2},
$$\frac{\delta_2}{\delta_1}\leq (1-\frac{1}{\alpha_2}),$$
and 
$$\frac{\delta_1}{\delta_0}\leq (1-\frac{1}{\alpha_1}),$$
respectively.

By multiplying the left hand and right hand of the above two inequalities separately,  
\begin{align*}
\delta_2 & =B(OPT)-B(X_2)\\
& \leq (1-\frac{C(X_2)-C(X_1)}{b_0c_0C(OPT)})(1-\frac{C(X_1)-C(X_0)}{b_0c_0C(OPT)})B(OPT)\\
& \leq e^{-\frac{C(X_2)-C(X_1)}{b_0c_0C(OPT)}}\times e^{-\frac{C(X_1)-C(X_0)}{b_0c_0C(OPT)}}B(OPT)\\
& =e^{-\frac{C(X_2)}{b_0c_0C(OPT)}}B(OPT).
\end{align*}

Thus,
$$B(X_2)\geq(1-e^{-\frac{C(X_2)}{b_0c_0C(OPT)}})B(OPT),$$
which implies
$$\frac{B(X_2)}{C(X_2)}\geq(1-e^{-\frac{C(X_2)}{b_0c_0C(OPT)}})\frac{b_0c_0C(OPT)}{C(X_2)}\frac{B(OPT)}{b_0c_0C(OPT)}.$$
Since $(1-e^{-x})x^{-1}$ is monotone decreasing in $x\leq \frac{1}{b_0c_0}$ and $\frac{C(X_2)}{b_0c_0C(OPT)}\leq \frac{1}{b_0c_0}$, 
\begin{flalign*}
\frac{B(X_2)}{C(X_2)} & \geq(1-e^{-\frac{C(X_2)}{b_0c_0C(OPT)}})\frac{b_0c_0C(OPT)}{C(X_2)}\frac{B(OPT)}{b_0c_0C(OPT)}\\
& \geq (1-e^{-\frac{1}{b_0c_0}})\frac{B(OPT)}{C(OPT)}.
\end{flalign*}
Finally, since $X_{2}=e_c$, 
$$\frac{B(e_c)}{C(e_c)}\geq (1-e^{-\frac{1}{b_0c_0}})\frac{B(OPT)}{C(OPT)}.$$
Hence, due to the estimation bias (Lemma \ref{lemma:merge}), {\algo} is a 
$(1-\epsilon)(1-e^{-\frac{1}{b_0c_0}})$--approximation algorithm for {\problem}. For simplicity, we combine the ratio as $(1-e^{-\frac{1}{b_0c_0}}-\epsilon)$, where $\epsilon$ is an arbitrarily small constant. Furthermore, when $b_0$ and $c_0$ are both bounded (i.e., 1), {\algo} is a constant approximation algorithm for {\problem}.
\end{proof}

\begin{timeC}


For the weighted directed graph $G=\{V,E\}$, let $|V|$ and $|E|$ denote the number of users and edges, respectively. \textblue{The time complexity of the ID phase is $O(M(|V|+|E|))$ since it first takes $O(2M|V|)$ time to construct the priority queue $Q$, and it then invests the budget in SCs and seeds by examining the MR of allocating an SC or activating a seed in $O(M|E|)$ and $O(M|V|)$ time, respectively. Note that both the size of coupon allocation and seed set are bounded by the budget. The time complexity of the GPI phase is $O(M|V|(|V|+|E|))$ since for each seed, the GPs are identified by a DFS-based traversal in $O(|V|+|E|)$ time, where the size of seed set is bounded by the budget. The time complexity of the SCM phase is $O(M|V||E|)$ because for each derived GP, the DIs are derived by retrieving SCs from each inactive user $v_j$ that is possibly to be influenced, whereas the number of GP is bounded by $|V|$, and the DIs are bounded by $|E|$. Moreover, the size of DIs is bounded by the vertex induced graph of GP. Therefore, the overall time complexity is $O(M|E|)+O(M|V||E|)+O(M|V||E|)=O(M|V||E|)$ and correlated to the size of the OSN. In particular, the time $O(M)$ of evaluating the expected benefit can be speeded up by Monte Carlo \cite{Kempe2003Maximizing} and reverse greedy methods \cite{tang2014influence}.}

\end{timeC}

\section{Experiment}
\label{sec:sim}

\subsection{Experiment Setup}
We compare {\algo} with IM \cite{Kempe2003Maximizing,chen2010scalable,chen2009efficient,goyal2011data,jung2012irie,leskovec2007cost,nguyen2016cost,nguyen2016stop,ohsaka2014fast,song2015influence,tang2018online,tang2017influence,tang2015influence,tang2014influence,zhou2013ublf} and PM \cite{tang2017profit,tang2018towards,zhu2018host,zhang2016profit,khan2016revenue,lu2012profit,zhu2013influence} in four OSNs, i.e., Facebook, Epinions, Google+ \cite{snapnets}, and Douban \cite{hung2016social}, with the detailed setting in Table \ref{tab:dataset}. Since IM and PM are not designed for SC, two real coupon strategies are adopted as follows. 1) Limited coupon strategy is provided by Dropbox, Airbnb, Booking.com, etc., where the SC constraint is specified by a constant $k$, i.e., $K_i=k,~\forall v_i \in V$. 2) Unlimited coupon strategy is provided by Uber, Lyft, and Hotels.com, etc., where the SC constraint of each user is specified by the number of friends, i.e., $K_i=N(v_i),~\forall v_i \in V$. For the limited coupon strategy, the SC constraint of each user is set to 32 according to Dropbox. We denote IM-U (PM-U) and IM-L (PM-L) as IM with unlimited coupon strategy and limited coupon strategy, respectively. For IM-U and IM-L, the seed size is set to $\frac{|V|}{2^n}$ for $n=0,1,\cdots,10$ \cite{tang2017profit}, and the seed size resulting in the maximum influence is selected as the result. The default setting of the IM algorithm follows the previous work \cite{tang2017profit}. \textblue{Moreover, we design a two-stage heuristic algorithm, IM-S.
The first stage employs the existing IM algorithm \cite{tang2017profit}. The second stage connects every two seeds with the shortest paths, where the weight of each edge $e(i,j)$ is $1-P(e(i,j))$ (i.e., an edge with a higher influence probability having a smaller weight). Afterward, IM-S uniformly distributes SCs to the users in the paths such that the overall seed cost and SC cost satisfy the investment budget constraint.}

The influence propagates under the extended IC model \cite{Kempe2003Maximizing} as described in Sec. \ref{sec:problem}. Following \cite{tang2017profit,tang2018towards,chen2010scalable,jung2012irie,nguyen2016cost,nguyen2016stop,tang2015influence,tang2014influence}, for each edge $e(i,j)$, the influence probability $P(e(i,j))$ is set to the reciprocal of $v_j$'s in-degree. We adopt the normal benefit setting \cite{tang2017profit}, i.e., for each user, the benefit is randomly generated by a normal distribution $\mathcal{N}(\mu,\sigma)$, where $\mu$ and $\sigma$ are the mean and standard deviation, respectively. For each user, the seed cost is proportional to the number of her friends (out-degree) \cite{tang2017profit}. Moreover, the uniform SC cost follows the real coupon strategies from Dropbox and Hotels.com. To evaluate the effect of different benefits and SC costs, we control the ratio $\lambda$ of the total benefit to the total SC cost, where $\lambda = \frac{\sum_{v_i \in V} b(v_i)}{\sum_{v_i \in V} c_{sc}(v_i)}$. Moreover, the effect of different seed costs and benefits is examined by controlling the ratio $\kappa$ of the total seed cost to the total benefit, where $\kappa = \frac{\sum_{\forall v_i \in V} c_{seed}(v_i)}{\sum_{\forall v_i \in V} b(v_i)}$. For all OSNs, the default setting of $\lambda$ and $\kappa$ are 1 and 10, respectively. We evaluate 1) the redemption rate, 2) total benefit, 3) the ratio of the total selected seed cost to total allocated SC cost (seed-SC rate), and 4) the average maximum hop number from seeds, by changing 1) $B_{inv}$, 2) coupon strategies of baseline algorithms, 3) $\lambda$, and 4) $\kappa$. Each simulation result is averaged over 1000 samples.

\begin{table}[t]
    \centering
        \caption{Datasets and the corresponding arguments}
    \begin{tabular}{ccccc}
    \hline
    Dataset & Facebook & Epinions & Google+ & Douban\\
    \hline
    Nodes & 4K & 76K & 108K & 5.5M\\       
    Edges & 88K & 509K & 13.7M & 86M\\
    \hline
    $B_{inv}$ & 10K & 50K & 200K & 1M\\
    $\mu,~\sigma$ & 10, 2 & 20, 4 & 50, 10 & 100, 20\\
    \hline
    \end{tabular}

    \label{tab:dataset}
\end{table}


\subsection{Simulation Results}
Fig. \ref{fig:Fig6} compares the investment efficiency of IM-U, PM-U, IM-L, PM-L, IM-S, and {\algo}, in the datasets of Table \ref{tab:dataset}. The results are similar for most datasets. Due to the space constraint, we present the results of Douban and Facebook. First, the impact on both redemption rate and total benefit from different investment budgets $B_{inv}$ is investigated in Fig. \ref{fig:Fig6}(a)--(b). The results manifest that {\algo} achieves the highest redemption rate and total benefit since it maximizes the total expected benefit while reducing the total cost. In Fig. \ref{fig:Fig6}(a), though the redemption rate of {\algo} sustains in a certain level as $B_{inv}$ increasing, Fig. \ref{fig:Fig6}(b) shows that the total benefit increases when $B_{inv}$ increases. Note that the total cost approximately equals to $B_{inv}$ for all algorithms in all settings and thereby is not presented in this section. 
\textblue{Fig. \ref{fig:Fig6} presents the results of IM-S. The redemption rate and total benefit of IM-S improve in Fig. \ref{fig:Fig6}(a) and (b) with a larger investment budget,
but IM-S acquires much smaller redemption rate and benefit compared with other approaches, because it only distributes SCs to the users on the shortest paths and thereby fails to obtain the benefits outside the paths. In Fig. \ref{fig:Fig6}(c) and (d), IM-S improves with a larger $\lambda$, which is the ratio of total benefit to total SC cost of users, because it is able to acquire a higher benefit from each SC with a higher $\lambda$. Fig. \ref{fig:Fig6}(e) and (f) manifest that IM-S incurs more running time (sometimes higher than {\algo}) as the budget grows, because more shortest paths spanning an increasing number of seeds become candidates for SC distribution.} 

Fig. \ref{fig:Fig7} investigates the impacts on the seed-SC rate with different $B_{inv}$, $\lambda$, and $\kappa$, in different datasets, and the results show that {\algo} carefully balances the investment in seeds and SCs according to different parameters. In Fig. \ref{fig:Fig7}(a)--(b), {\algo} increases the investment in seeds (i.e., the seed-SC rate increases) while $B_{inv}$ increasing. With the generous in investment, more seeds can be deployed for more influential sources. In Fig. \ref{fig:Fig7}(c)--(d), the benefit of each user increases while the SC cost is fixed, which implies that the benefit per unit of investment also increases. Thus, for the seeds, the benefit surpasses its seed cost and the benefit of allocating an SC to other non-seed users, {\algo} decides to invest more in seeds. Moreover, the seed-SC rate of baseline algorithms has negligible changes in Fig. \ref{fig:Fig7}(a)--(d) since they do not consider the SC allocation. In Fig. \ref{fig:Fig7}(e)--(f), for different $\kappa$, all baseline algorithms increase the investment in seeds since the seed cost increases as $\kappa$ increasing. However, {\algo} decreases the investment in seeds while the seed cost ($\kappa$) increases and invests more in SCs to optimize the redemption rate. Thus, {\algo} is capable of balancing the investment in seeds and SCs according to different budgets, $\lambda$, and $\kappa$.


Table \ref{tab:hop} and \ref{tab:runningTime} show the average farthest hop from seeds and the average running time, respectively. Table \ref{tab:hop} shows that {\algo} can deepen the influence spread by allocating SCs and the average farthest hops ranging from $2.046$ to $3.355$. However, for the baseline algorithms, the average farthest hops reside between $1$ to $1.939$. Thus, the results manifest the effectiveness of {\algo} to disseminate SCs to a wider area in OSNs. Moreover, the running time is listed in Table \ref{tab:runningTime}. Since {\algo} activates the seeds and allocates the SCs under the investment budget, the running time is closely related to the investment budget. Thus, the running time is proportional to the investment budget and less related to the size of OSN, which is shown in Table \ref{tab:runningTime}.


\begin{figure}[t]
    \centering
    \includegraphics[height=4.3in,width=3.7in]{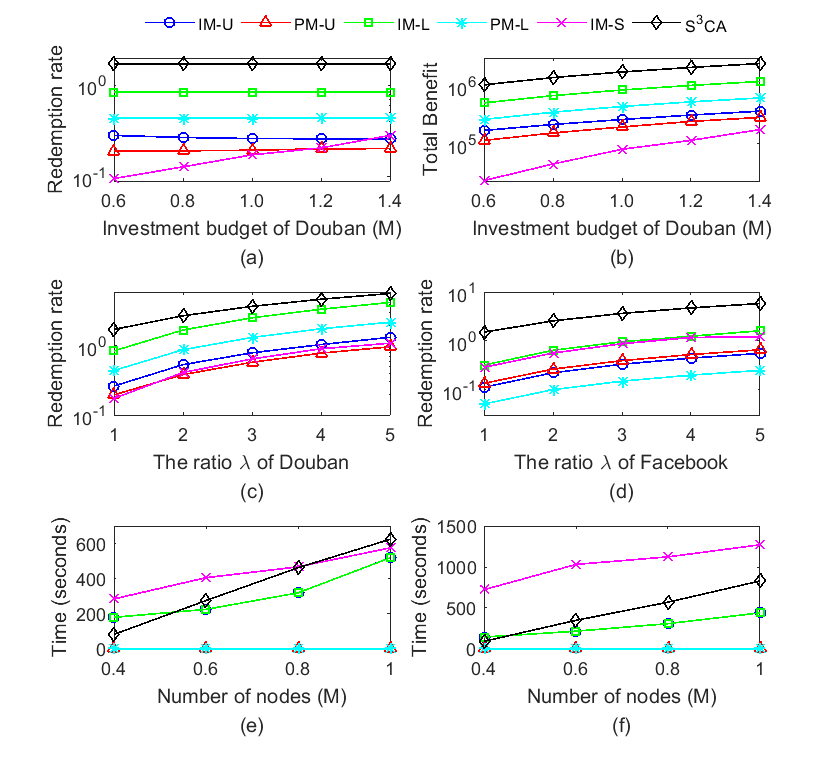}
    \caption{Investment efficiency. (a) Redemption rate with different investment budgets in Douban. (b) Total benefit with different investment budgets in Douban. (c) Redemption rate with different $\lambda$ in Douban. (d) Redemption rate with different $\lambda$ in Facebook. (e) Running time with $B_{inv}=2$M. (f) Running time with $B_{inv}=3$M}
    \label{fig:Fig6}
\end{figure}

\begin{figure}[t]
    \centering
    \includegraphics[height=3.8in,width=3.7in]{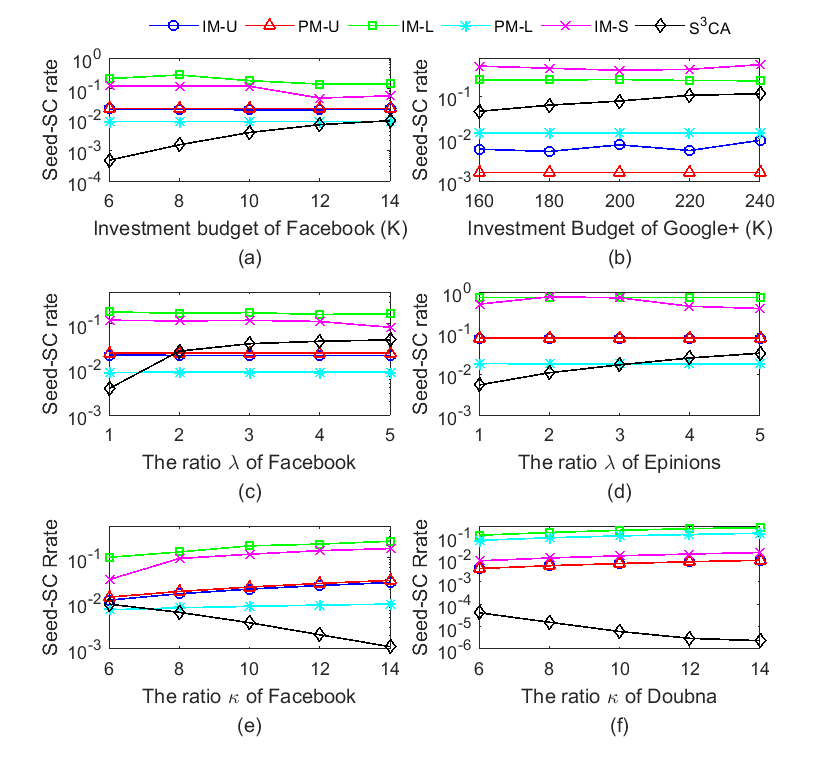}
    \caption{Seed-SC rate. (a) Different investment budgets in Facebook. (b) Different investment budgets in Epinions. (c) Different $\lambda$ in Facebook. (d) Different $\lambda$ in Google+. (e) Different $\kappa$ in Facebook. (f) Different $\kappa$ in Douban.}
    \label{fig:Fig7}
\end{figure}

\textblue{
\subsection{Case Study}
We have incorporated real SC policies provided by Airbnb and Booking.com on the real dataset in Table \ref{tab:dataset}. We have employed the adoption model \cite{tang2018stochastic} for each user to find the users that accept SCs. We have adopted real gross margin to set the benefit according to the accounting research \cite{kieso2010intermediate}. We have used the real datasets of social networks in \cite{hung2016social}. More specifically, the SC costs are 50 and 100, and the SC allocations are 100 and 10 according to Airbnb and Booking.com, respectively. Since Booking.com has not revealed its SC cost, here we refer to the one in Hotels.com to assign the SC cost. Moreover, the adoption model \cite{tang2018stochastic}, which quantifies the probability of users adopting a coupon, is to uniformly select 85\%, 10\%, and 5\% of users with $\sqrt[3]{c_{sc}}$, $c_{sc}$, $c_{sc}^2$ and all normalized by $\sqrt[3]{c_{sc}}+c_{sc}+c_{sc}^2$. For the benefit setting, we adopt the gross margin for SCs from accounting research \cite{kieso2010intermediate}, which is defined as $\frac{b(v_i)-c_{sc}(v_i)}{b(v_i)}*100$ (\%) for each $v_i$.

Fig. \ref{fig:FigCase} presents the new results with the above setting. The redemption rate increases with a larger gross margin in Fig. \ref{fig:FigCase}(a) due to more benefits generated from each user. The redemption rate of Booking.com in Fig. \ref{fig:FigCase}(c) is higher because Airbnb has higher SC allocation, but more SCs are not redeemed by users. Therefore, fewer users are influenced and activated. Moreover, for PM-L and PM-U, the redemption rate increases with 60\% gross margin in Fig. \ref{fig:FigCase}(c), because the benefit is sustained, but some seeds can be discarded to reduce the cost as shown in Fig. \ref{fig:FigCase}(d). {\algo} achieves the highest redemption rate with different gross margins in Airbnb and Booking.com, since in addition to finding influential seeds, {\algo} obtains more benefits with guaranteed paths, which connect and influence high-benefit users by SCs.
}

\begin{figure}[t]
    \centering
    \includegraphics[height=2.5in,width=3.7in]{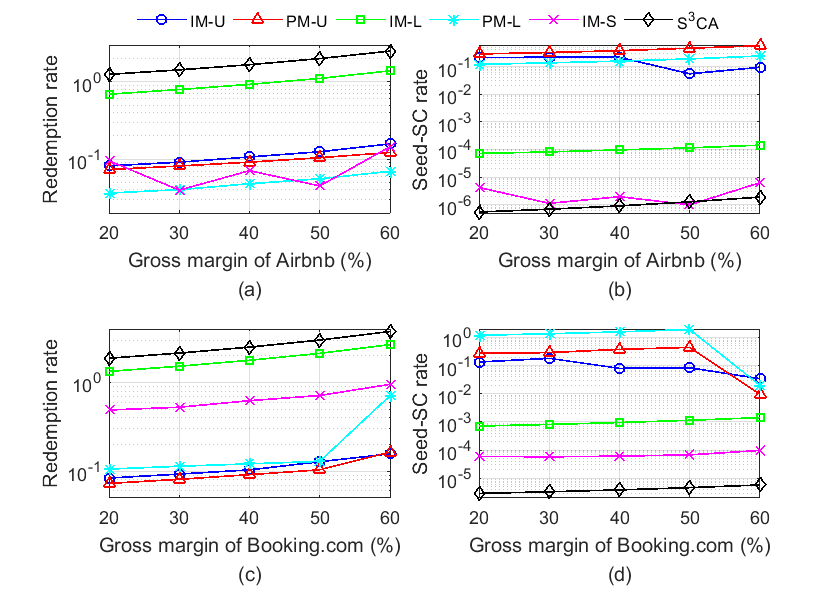}
    \caption{Case study with different gross margin. (a) Redemption rate of Airbnb. (b) Seed-SC cost of Airbnb. (c) Redemption rate of Hotels.com. (d) Seed-SC cost of Hotels.com.}
    \label{fig:FigCase}
\end{figure}

\textblue{
\subsection{Performance of {\algo}}
To validate the approximation ratio, we have compared {\algo} with the optimal solution obtained by computation-intensive exhaustive search in small networks with 150 nodes generated by PPGG \cite{shuai2013pattern}.\footnote{The input parameters for PPGG are 1) 11 patterns, 2) support of 1000, 3) clustering coefficient of 0.6394, and 4) power-law parameter $\eta = 1.7$ and $2.5$.} We change the gross margin for SCs from accounting research \cite{kieso2010intermediate}. Compared with other algorithms, Fig. \ref{fig:FigOPT}(a) shows that {\algo} is closer to the optimal solution, and the redemption rates of some baseline algorithms are even smaller than the worst-case bounds of {\algo}, which are derived from the optimal solutions multiplied by the approximation ratio. Fig. \ref{fig:FigOPT}(b) shows that all solutions returned by {\algo} are above the worst-case bound, indicating that the approximation ratio holds for the empirical cases.

To test the scalability, we first adopt PPGG \cite{shuai2013pattern} to generate large Facebook-like synthetic networks. Fig. \ref{fig:FigScalability} presents the running time and the explored ratio (i.e., the ratio of the number of explored nodes of {\algo}  to the network size) with different network sizes and investment budgets. Fig. \ref{fig:FigScalability}(a) indicates that the running time becomes larger as the network size grows, but the explored ratio decreases under a fixed budget in Fig. \ref{fig:FigScalability}(b) (i.e., {\algo}  stops exploring new nodes when the budget is run out). 
In contrast, Fig. \ref{fig:FigScalability}(c) and (d) indicate that both the running time and explored ratio become larger as the investment budget increases, since {\algo} in this case is required to examine more candidates for the efficient distribution of SCs.
}

\begin{figure}[t]
    \centering
    \includegraphics[height=2.5in,width=3.7in]{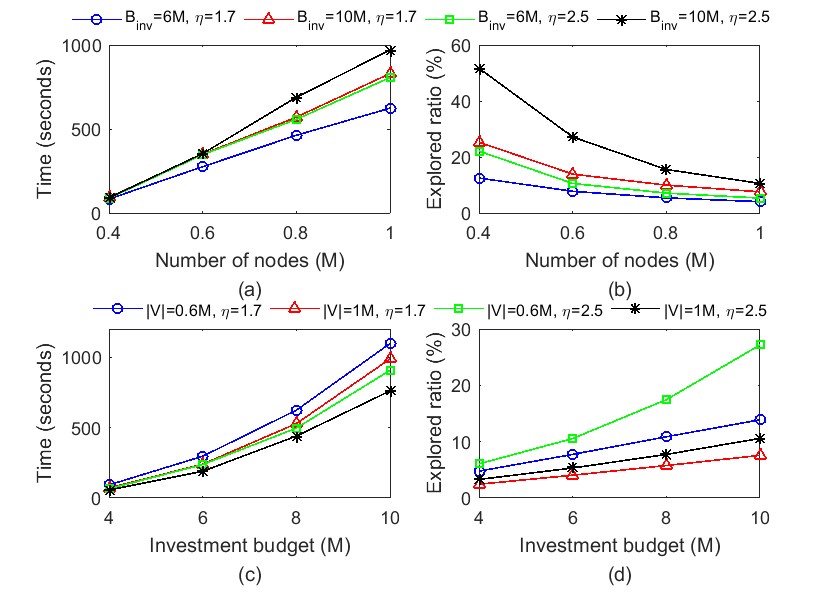}
    \caption{Scalability experiments. (a) Running time with different network size. (b) Explored ratio with different network size. (c) Running time with different investment budgets. (d) Explored ratio with different investment budgets.}
    \label{fig:FigScalability}
\end{figure}

\begin{figure}[t]
    \centering
    \includegraphics[height=1.3in,width=3.5in]{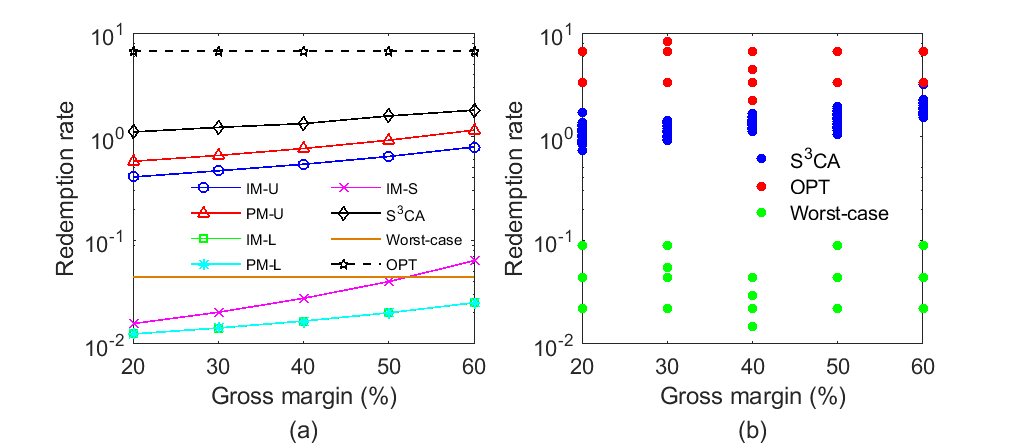}
    \caption{Performance of {\algo}. (a) Average results of baselines, {\algo}, OPT, and worst-case. (b) All results of {\algo}, OPT, and worst-case.}
    \label{fig:FigOPT}
\end{figure}



\begin{table}[t]
    \centering
 \caption{Average farthest hops from seeds}
    \begin{tabular}{cccccc}
    \hline
    Dataset & IM-U & IM-L & PM-U & PM-L & {\algo}  \\
    \hline
    Facebook & 1.958 & 1.000 & 1.872 & 1.004 & 3.579 \\
    Epinions & 1.381 & 1.000 & 1.120 & 1.002 & 3.363 \\
    Google+ & 1.724 & 1.000 & 1.723 & 1.000 & 2.690 \\
    Douban & 1.014 & 1.000 & 1.002 & 1.001 & 3.134 \\
    \hline
    \end{tabular}

    \label{tab:hop}
\end{table}

\begin{table}[t]
    \centering
    \caption{Average running time of {\algo}}
    \label{tab:runningTime}
    \begin{tabular}{cccccc}
    \hline
    $B_{inv}$ (K) & 6 & 8 & 10 & 12 & 14  \\
    \hline
    Facebook (seconds) & 1.02 & 2.24 & 3.65 & 5.49 & 7.70 \\
    \hline
    $B_{inv}$ (K) & 30 & 40 & 50 & 60 & 70  \\
    \hline
    Epinions (seconds) & 6.51 & 10.80 & 15.20 & 20.66 & 26.52 \\
    \hline
    $B_{inv}$ (M) & 0.6 & 0.8 & 1 & 1.2 & 1.4  \\
    \hline
    Douban (seconds) & 121.94 & 174.05 & 242.12 & 304.05 & 378.45 \\
    \hline
    $B_{inv}$ (K) & 60 & 80 & 100 & 120 & 140  \\
    \hline
    Google+ (seconds) & 821.29 & 928.76 & 1059.93 & 1130.95 & 1282.55 \\
    \hline
    \end{tabular}
\end{table}

\section{Conclusion}
\label{sec:conclusion}

To the best of our knowledge, this paper makes the first attempt to explore the seed selection with SC allocation under a limited investment budget. We formulate a novel optimization problem {\problem} and design an approximation algorithm {\algo} to optimize the redemption rate in social coupon scenario. {\algo} first deploys the investment in seeds and social coupons by carefully examining the marginal redemption of three strategies which are deepening the influence spread, broadening the influence spread, and initiating another start of the influence spread. {\algo} then identifies the guaranteed paths to explore the opportunity to optimize the redemption rate by maneuvering the social coupons. After obtaining the guaranteed paths, {\algo} further examines the amelioration index and the deterioration index of the amelioration and the deterioration by maneuvering social coupons to users and retrieving social coupons from users, respectively. Finally, {\algo} optimizes the redemption rate by maneuvering the SCs from the users with lower deterioration indices to the users with higher amelioration index. Simulation results show that {\algo} can effectively improve the investment efficiency up to 30 times of the results of the baseline algorithms.





\bibliographystyle{ieeetr}
\bibliography{bibliographies/ref}


\end{document}